\documentclass{article}
\usepackage{float}
\usepackage[utf8]{inputenc}   
\usepackage[english]{babel}      
\usepackage{amsmath, amssymb, amsthm, mathrsfs}     
\usepackage[a4paper, margin=3cm]{geometry} 
\usepackage{graphicx}         
\usepackage{natbib}
\usepackage{hyperref}         
\usepackage{tikz}
\allowdisplaybreaks
\usepackage{float} 
\usepackage{graphicx, pgfplots, xcolor}
\usepackage{bm}
\usepackage{adjustbox}
\definecolor{granate}{RGB}{128, 0, 64} 
\usepackage[T1]{fontenc}
\usepackage{subcaption}
\usepackage{booktabs}
\usepackage{amsthm}
\usepackage{xcolor}
\newtheorem{theorem}{Theorem}

\newtheorem{definition}[theorem]{Definition}

\newtheorem{proposition}[theorem]{Proposition}
\newtheorem{remark}[theorem]{Remark}
\newcommand{\newproof}[2]{%
  \newenvironment{#1}{\begin{proof}[#2]}{\end{proof}}%
}
\newproof{pot1}{Proof of Theorem \ref{th:theo_1}}
\newproof{pot2}{Proof of Theorem \ref{th:theo_2}}
\newproof{pot3}{Proof of Theorem \ref{th:theo_3}}
\newproof{pot4}{Proof of Theorem \ref{th:theo_4}}
\newproof{pot5}{Proof of Theorem \ref{th:theo_5}}
\newproof{pot6}{Proof of Theorem \ref{th:theo_6}}

\begin{document}

\title{\textbf{Robust Tests for Step-Stress Models under Exponential Lifetimes}}
\author{Mar\'{i}a Jaenada$^{(1)}, $ Juan Manuel Mill\'{a}n$^{(2)}$ and Leandro Pardo$^{(2)}$ \\
{\small $^{(1)}$ Department of Statistics, O.R. and N.A., UNED, Madrid, Spain}\\
{\small $^{(2)}$ Department of Statistics and O.R., Complutense University of Madrid, Spain}
}
\date{}
\maketitle
\begin{abstract}
Highly reliable products with extended lifetimes present a  challenge in reliability analysis: obtaining enough failure data under normal operating conditions is often incompatible with reasonable time  and cost constraints. Step-Stress Accelerated Life Tests (SSALTs) offer  a practical solution by progressively increasing stress levels to  accelerate product degradation, allowing results to be extrapolated to  normal conditions. However, the small sample sizes and Type-I censoring  inherent to these experiments render classical maximum likelihood-based inference  vulnerable to data contamination. 
Robust point estimation has been studied in the literature. However, robust test statistics have not yet been developed in this context. In this paper, we propose robust test statistics for SSALTs under exponential lifetime distributions, based on the minimum density power divergence estimator (MDPDE). 
This paper directly works with the exact failure times recorded before the end of the experiment. Once this limit time is reached, all surviving components are censored. This creates a mixed discrete-continuous distribution, which   is a major analytical breakthrough in the context of SSALT. 
We  introduce the restricted version of the MDPDE, establish its  asymptotic properties, and construct Z-type and Rao-type  test statistics for linear hypotheses on the model parameters. The proposed  tests are shown to maintain their nominal significance levels and statistical  power under data contamination, where classical MLE-based procedures fail. An  extensive  simulation study confirms the robustness gains of the  proposed methods, and a real data application illustrates their  practical value.
\end{abstract}

\section{Introduction}
Reliability analyses aim to infer the lifetime behavior of a device or system, that is, the time until failure under normal operating conditions. Many modern products and systems are designed with long mean times to failure (MTTF), often extending over several years. Under these circumstances, obtaining sufficient lifetime observations for statistical inference under operating conditions becomes challenging, making it difficult to draw conclusions about the lifetime behavior of the device (\citep{Nelson2004}).
To address this challenge, accelerated life tests (ALTs) are applied in industry by subjecting devices to higher-than-normal physical conditions that wear down the product and reduce its time to failure. For example, experimenters can increase environmental conditions such as temperature, voltage, pressure, load, humidity, or some combination to induce deterioration \citep{Nelson1990, Escobar2006}. ALT models require a relationship between lifetime behavior and stress load, allowing the extrapolation of results obtained under accelerated conditions to normal operating conditions. There exist mainly three types of ALTs designs used in practice, depending on how the stress is applied: Constant Stress ALTs (CSALTs), Step-Stress ALTs (SSALTs) and Progressive Stress ALTs (PSALTs). Among these,  \citet{Nelson1990} demonstrates that SSALTs can provide more accurate estimates than CSALTs under similar experimental conditions and \citet{Alhadeed2005} further highlights their advantages in terms of shorter experimental duration and reduced costs.

In SSALTs, the total experimental duration is divided into stress intervals defined by fixed time points $\tau_0 = 0 < \tau_1 < \tau_2$. During each interval, a different stress level is applied to the $N$ test units, with the stress level assumed to increase in discrete steps  to induce failure. The experiment concludes at a fixed time $\tau_2$, beyond which surviving units are right-censored. This design corresponds to a Type-I censoring scheme, which is particularly appealing in practice as it guarantees a predetermined experimental duration. Under this setup, the Cumulative Exposure (CE) model \citep{Sedyakin1966, Nelson1980} is commonly employed to relate the lifetime distribution under each stress level to the overall observed failure times, assuming that the remaining lifetime of a unit depends only on its current cumulative exposure to stress, regardless of how that exposure was accumulated.

A common lifetime distribution in SSALTs is the exponential distribution, which, despite its simplicity, remains widely used due to its analytical tractability and its natural interpretation in terms of a constant hazard rate. Inference for step-stress models under exponential lifetimes has been extensively studied in the literature; see, for instance, \citet{Balakrishnan2009} for exact inference under Type-I hybrid censoring, \citet{Miller1983, DeGroot1979} for foundational results, and \citet{Xiong1998} for inference under Type-I censoring.

Statistical inference for these models is typically carried out via the MLE, which possesses desirable asymptotic properties such as consistency and asymptotic normality. Moreover, widely used hypothesis tests such as the Wald test, the Rao score test, and the likelihood ratio test are based on the MLE or its restricted version (\citep{Rao1948, Wald1943}). However, it is well known that the MLE lacks robustness and can be severely affected by data contamination or the presence of outliers, leading to unreliable estimates and distorted test decisions. In reliability testing, where a small number of anomalous observations can arise from measurement errors, faulty units, or unexpected failure modes, this sensitivity is a serious practical concern.
To overcome this limitation, robust inferential methods based on divergence measures have been developed. In particular, the MDPDE, introduced by \citet{Basu1998}, provides a flexible family of estimators indexed by a tuning parameter $\beta \geq 0$ that governs the trade-off between robustness and efficiency: for $\beta = 0$ the MDPDE reduces to the MLE, while increasing $\beta$ yields estimators that progressively downweight the influence of outlying observations. This approach has been extended to the step-stress setting by \citet{Balakrishnan2023a}, who introduced robust MDPDEs for SSALTs under interval censoring and exponential lifetimes. The restricted version of this estimator was studied  in \citet{Balakrishnan2023c} to construct robust Rao-type tests. Robust MDPDEs under Weibull, Gamma, and Lognormal lifetimes have been discussed in \citet{Balakrishnan2023b, Balakrishnan2024a,Balakrishnan2024b,Balakrishnan2025,Balakrishnan2026a,Balakrishnan2026b}.

In this paper, we work directly  with the exact  failure times recorded before the end of the experiment. 
Once the termination time is reached, all surviving components are right-censored. This gives rise to a mixed discrete-continuous model: one part of the sample consists of continuous observations corresponding to exact failure times, whereas the remaining observations contribute a discrete probability mass at the censoring time. Developing a robust statistical framework for this mixed discrete–continuous setting represents a major analytical breakthrough compared with the fully continuous or fully discrete cases.
The mixed continuous-discrete model was initially proposed in \citet{jaenada2025}, where robust point estimation with exponential and Weibull distributions was considered. In this paper, we consider the problem of robust testing for SSALTs under exponential lifetime distributions and Type-I censoring. Specifically, we construct robust Z-type, Wald-type, and Rao score-type test statistics based on the MDPDE and its restricted version (RMDPDE), and we study their asymptotic distributions under the null hypothesis as well as their behavior under data contamination. The theoretical results are complemented by an extensive Monte Carlo simulation study that illustrates the robustness gains over the classical MLE-based tests.

The rest of this paper is organized as follows. Section~\ref{sec:sec2} presents the step-stress model and the CE model under exponential lifetimes and Type-I censoring. Section~\ref{sec:sec3} introduces the restricted version of the MDPDE along with its asymptotic properties. Sections~\ref{sec:sec4} and \ref{sec:sec6}  develops robust Z-type, and Rao test statistics. Section~\ref{sec:sec7} presents the simulation study. Section~\ref{sec:sec8} presents a real data analysis. In Section~\ref{sec:sec9} some concluding remarks are presented.

\section{The minimum density power divergence estimator} \label{sec:sec2}

\bigskip

%
Before defining the MDPDE, let us justify the use of a mixed distribution in this context. The standard Step-Stress Accelerated Life Testing (SSALT) scheme with two stress levels is structured as follows. 

Initially, a total of $N$ identical units are placed on test under an initial stress level $x_1$ during the time interval $(0, \tau_1]$. During this first stage, $n_1$ units fail, and their exact failure times are recorded as the ordered statistics $t_{1:N}, \dots, t_{n_1:N}$. At the first censoring threshold $\tau_1$, the stress level is increased to $x_2$ for the remaining $N - n_1$ surviving units. This second stage of the experiment runs during the interval $(\tau_1, \tau_2]$, where an additional $n_2$ units fail, yielding the exact failure times recorded as $t_{n_1+1:N}, \dots, t_{n_1+n_2:N}$. Finally, at the end of the test (the terminal censoring time $\tau_2$), the experiment is terminated, leaving $N - n_1 - n_2$ surviving components whose exact failure times are right-censored at $\tau_2$.

Under an exponential lifetime assumption, the applied stress levels $x_1$ and $x_2$ scale the units' lifetimes through their respective scale parameters $\lambda_1$ and $\lambda_2$. Consequently, the scale parameter vector is denoted by $\bm{\lambda} = (\lambda_1, \lambda_2)^T$. In general, for a constant stress level, the cumulative distribution function cdf is given by $F(t) = 1 - \exp(-t/\lambda)$. However, under the cumulative exposure model (CEM) framework with two stress steps, the cumulative distribution function of the lifetime of a unit, denoted by $F^*(t|\bm{\lambda})$, is expressed as:

\begin{equation}
F^*(t|\bm{\lambda})=
\begin{cases}
0, & t < 0, \\
F^*_{1}(t|\bm{\lambda}) = 1-\exp\left(-\frac{t}{\lambda_1} \right), & 0 \leq t < \tau_1, \\
F^*_{2}(t + h|\bm{\lambda}) = 1-\exp\left(-\frac{1}{\lambda_2}\left(t + \frac{\lambda_2}{\lambda_1}\tau_1 - \tau_1\right)\right), & \tau_1 \leq t < \infty.
\end{cases}
\label{eq:funFstar}
\end{equation}

The time-shift parameter $h$ represents the equivalent exposure time and is introduced to guarantee the continuity of the distribution function at the stress-change point $\tau_1$, ensuring that $F^*_1(\tau_1|\bm{\lambda}) = F^*_2(\tau_1+h|\bm{\lambda})$. Solving this boundary condition yields:
\begin{equation}
h = \frac{\lambda_2}{\lambda_1}\tau_1 - \tau_1.
\label{eq:h_factor}
\end{equation}

This piecewise structure highlights the mixed nature of the distribution when censoring is applied. While exact failure times $t_{i:N}$ are observed continuously on the intervals $(0, \tau_1]$ and $(\tau_1, \tau_2]$, any unit surviving beyond $\tau_2$ contributes to a discrete probability mass of size $1 - F^*(\tau_2|\bm{\lambda})$ concentrated precisely at the censoring point $\tau_2$.
While the units under test are subject to continuous monitoring throughout the test, the data available upon the study's conclusion are partitioned: exact failure times are recorded for some, but for the surviving units, only their surviving status is observed. 
As a result, the observable cdf, denoted by $F_T$, is characterized as a mixed distribution. It comprises a continuous component over the interval $(0, \tau_2)$ and a discrete point mass at time $\tau_2$ equivalent to $1 - F^*_2(\tau_2+h|\bm{\lambda})$. This mass represents the survival probability at the experiment's termination time. Specifically, under the assumption of exponential lifetimes, the observable cdf is defined as
\begin{equation}
F_T(t|\bm{\lambda})=
\begin{cases}
0 & \quad t<0
\\
F_{1}(t|\bm{\lambda})=1-\exp\left(-\frac{t}{\lambda_1} \right) \quad& 0\leq t <\tau_1 
\\
F_{2}(t+h|\bm{\lambda})=1-\exp\left(-\frac{1}{\lambda_2}\left(t+\frac{\lambda_1}{\lambda_2}\tau_1-\tau_1\right)\right) \quad &\tau_1 \leq t < \tau_2
\\
1 \quad & \tau_2 \leq t,
\end{cases}
\label{eq:funF}
\end{equation}
where $\bm{\lambda} =\left(\lambda_1, \lambda_2 \right)$ are the exponential scale parameters under constant stress levels $x_1$ and $x_2$.

By a similar reasoning, the true underlying cdf $G_{T}$ of the lifetime truncated at the experiment endpoint, is represented as a piecewise distribution of the form
\begin{equation}
	G_T(t)=
	\begin{cases}
		0 \quad & t<0
		\\
		G_1(t) \quad & 0 \leq t <\tau_1
		\\
		G_2(t) \quad & \tau_1 \leq t < \tau_2
		\\
		1 \quad & \tau_2 \leq t,
	\end{cases}
\label{eq:funG}
\end{equation}
where $G_1(\cdot)$ and $G_2(\cdot)$ denote the cdfs under constant stress levels $x_1$ and $x_2$, respectively.
Note that, rather than assuming a specific parametric model to relate the cdf across different stress conditions, we simply allow the distribution to vary with changes in the stress level. 

For the sake of notation simplicity, we define
\begin{align*}
	f_1(t|\bm{\lambda})=\frac{\partial F_1(t|\bm{\lambda})}{\partial t} \quad \quad g_1(t)=\frac{\partial G_1(t)}{\partial t} \quad & 0 < t <\tau_1 
	\\
	f_2(t+h|\bm{\lambda})=\frac{\partial F_2(t+h|\bm{\lambda})}{\partial t} \quad \quad g_2(t)=\frac{\partial G_2(t)}{\partial t} \quad & \tau_1 < t <\tau_2 .
\end{align*}
These functions are not proper density functions themselves, but are proportional to them up to a constant factor, as they represent the continuous component of the mixed distribution.

Given the failure times \( t_{1:N}, \dots, t_{n_1:N}, t_{n_1+1:N}, \dots, t_{n_1+n_2:N} \), the log-likelihood of the SSALT model 
is given by

\begin{equation*}
	\begin{aligned}
	\ell(\bm{\lambda}) =&\log\left(N!\right) -\log \left((N-n_1-n_2)!\right)\\
	& + \sum_{i=1}^{n_1} \log\left(g_1(t_{i:N}|\bm{\lambda})\right) \\
	& + \sum_{i=n_1+1}^{n_1+n_2} \log\left(g_2(t_{i:N} + h|\bm{\lambda})\right) \\
	& +(N-n_1-n_2)\log\left( 1 - G_2(\tau_2 + h|\bm{\lambda}) \right),
	\end{aligned}
	\label{eq:likelihood}
\end{equation*}
and the MLE for the SSALT is defined as
\begin{equation}
\left(\hat{\lambda}_1^{MLE}, \hat{\lambda}_2^{MLE} \right) = \arg \max_{ \bm{\lambda} \in \mathbb{R}^{+} \times \mathbb{R}^{+}} \ell (\bm{\lambda}).
\end{equation}
It is worth noting that the MLEs of $\lambda_1$ and $\lambda_2$ exist simultaneously if and only if at least one failure is observed in each stress interval. This condition is assumed to hold throughout the paper.

The MDPDE is based on an estimation procedure that consists in minimizing a statistical distance between the empirical data distribution and the parametric model. To motivate this framework, we first introduce the Kullback-Leibler (KL) divergence. This introduction is mathematically justified by the fact that minimizing the KL divergence with respect to the parameter vector $\bm{\lambda}$ is equivalent to maximizing the log-likelihood function of the sample. Consequently, the MLE can be formally obtained as a minimum divergence estimator under the KL divergence. This crucial link allows us to view classical estimation through the theory of divergence minimization, serving as the baseline to subsequently introduce robust generalizations such as the MDPDE.

From the above, we can define the KL divergence between the mixed distributions $F_T(\cdot|\bm{\lambda})$ and $G_T(\cdot)$, as defined in \eqref{eq:funF} and \eqref{eq:funG}, is given by,
\begin{equation}
	\begin{aligned}
		d_{KL}\left(G_T\left(\cdot\right), F_T\left(\cdot|\bm{\lambda}\right)\right)=&\int_0^{\tau_1} g_1(t) \log\left(\frac{g_1(t)}{f_1(t|\bm{\lambda})}\right) dt +\int_{\tau_1}^{\tau_2}g_2(t)\log\left(\frac{g_2(t)}{f_2( t+h|\bm{\lambda})}\right)dt\\
		&+\left(1-G_2\left(\tau_2\right)\right)\log\left(\frac{1-G_2(\tau_2)}{1-F_2(\tau_2+h|\bm{\lambda})}\right).
	\end{aligned}
\label{eq:kl_def}
\end{equation}
The KL divergence for mixed distributions is defined by combining the contributions of their continuous and discrete components, following an approach previously proposed in the context of Shannon entropy ( \citet{Nair2007}).


The optimal estimate of $\bm{\lambda}$ based on the KL divergence is such that the assumed parametric model fits the true distribution as closely as possible. Since the true distribution functions $G_1(\cdot)$ and $G_2(\cdot)$ are unknown, given a random sample from the experiment $(t_{1:N}, \dots, t_{n_1+n_2:N})$, we approximate them using their respective empirical distribution functions. Consequently, we obtain
\begin{equation}
	\begin{aligned}
		d_{KL}\left(G_T\left(\cdot\right), F_T\left(\cdot|\bm{\lambda}\right)\right)\approx&-\frac{1}{N}\sum_{i=1}^{n_1} \log(f_1(t_{i:N}|\bm{\lambda}))-\frac{1}{N}\sum_{i=n_1+1}^{n_1+n_2}\log\left(f_2(t_{i:N}+h|\bm{\lambda})\right)
		\\
		&-\frac{N-n_1-n_2}{N}\log\left(1-F_2(\tau_2+h|\bm{\lambda}\right)+ \text{Constant Terms}.
	\end{aligned}
\end{equation}

The previous expression is proportional, up to an additive constant, to the log-likelihood function. 
The estimator based on the KL divergence will be the one that minimizes such divergence, and therefore, the estimator is given by
\begin{align*}
\left(\hat{\lambda}_1^{KL}, \hat{\lambda}_2^{KL} \right) = \arg \min_{ \bm{\lambda} \in \mathbb{R}^{+} \times \mathbb{R}^{+}} d_{KL}\left(G_T\left(\cdot\right),F_T(\cdot|\bm{\lambda})\right)
\end{align*}
As is well known, the KL divergence yields the MLE within the minimum-distance framework.
%


The Density Power Divergence (DPD) is a widely adopted divergence measure for developing robust statistical inference methods from the minimum distance framework. It facilitates a balanced trade-off between efficiency and robustness, controlled by a tuning parameter $\beta$. Originally introduced by \citet{Basu1998}. The DPD has gained extensive use due to its capacity to produce estimators that maintain high efficiency under the model while significantly enhancing robustness against outliers.

Using similar arguments to those used above for mixed distributions, we define the DPD between $G_T(\cdot)$ and $F_T(\cdot|\bm{\lambda})$ as a combination of the divergences between their continuous and discrete components as follows
\begin{equation}
	\begin{aligned}
		d_{\beta}\left(G_T(\cdot), F_T(\cdot|\bm{\lambda})\right) &=\int_0^{\tau_1} \left(f_1(t|\bm{\lambda})^{\beta+1} - \left(1+\frac{1}{\beta} \right) f_1(t|\bm{\lambda})^{\beta} g_{1}(t)+\frac{1}{\beta}g_{1}(t)^{\beta+1} \right) dt
		\\
		&+\int_{\tau_1}^{\tau_2} \left(f_2(t+h|\bm{\lambda})^{\beta+1} - \left(1+\frac{1}{\beta} \right) f_2(t+h|\bm{\lambda})^{\beta} g_{2}(t)+\frac{1}{\beta}g_{2}(t)^{\beta+1} \right)dt
		\\
		&+\left(1-F_2(\tau_2+h|\bm{\lambda})\right)^{\beta+1}-\left(1+\frac{1}{\beta}\right)\left(1-F_2(\tau_2+h|\bm{\lambda})^{\beta}\right)\left(1-G_{2}(\tau_2)\right)
		\\
		&+\frac{1}{\beta}\left(1-G_{2}(\tau_2)\right)^{\beta+1}.
	\end{aligned}
\end{equation}
Now, given a random sample from the experiment, $(t_{1:N}, \dots, t_{n_1+n_2:N})$, and replacing the true (unknown) cdf with its empirical estimate, the empirical DPD is given, up to an additive constant, by
\begin{equation}
	\begin{aligned}
	H_{N}^{\beta}\left(G_T(\cdot), F_T(\cdot|\bm{\lambda})\right)& =  \int_{0}^{\tau_1}f_1(t|\bm{\lambda})^{\beta+1}dt +\int_{\tau_1}^{\tau_2}f_2(t+h|\bm{\lambda})^{\beta+1}dt+\left(1-F_2(\tau_2+h|\bm{\lambda})\right)^{\beta+1}
	\\
	&-\frac{\beta+1}{\beta N}\left\{ \sum_{i=1}^{n_1}f_1(t_{i:N}|\bm{\lambda})^\beta +\sum_{i=n_1+1}^{n_1+n_2}f_2(t_{i:N}+h|\bm{\lambda})^\beta \right.
	\\
	&\left.+ (N-n_1-n_2)\left(1-F_2(\tau_2+h|\bm{\lambda})\right)^{\beta}\right\}+\frac{1}{\beta}. 
	\end{aligned}
\end{equation}
\par
Consequently, the MDPDE for a SSALT model
for a fixed tuning parameter $\beta > 0$, is defined as
\begin{equation}
(\hat{\lambda}_1^{\beta}, \hat{\lambda}_2^{\beta})=\arg \min_{\bm{\lambda} \in \mathbb{R}^{+} \times \mathbb{R}^+} H_{N}^{\beta}(\bm{\lambda}).
\label{eq:dpd_min}
\end{equation}

Let us now derive explicit expression of the DPD-based loss for exponential lifetimes.
\begin{proposition}\label{prop:prop1}
Under exponential lifetimes with scale parameters $\bm{\lambda} = (\lambda_1, \lambda_2)$ as given in \eqref{eq:funF}, the DPD-based loss $H_{N}^{\beta}(\bm{\lambda})$ is given by
    \begin{equation}
    H_{N}^{\beta}(\bm{\lambda})=h_1(\bm{\lambda})+h_2(\bm{\lambda}),
    \end{equation}
    with
    \begin{align*}
    	h_1(\bm{\lambda})&=\frac{1}{\lambda_1^{\beta}(\beta+1) } - \frac{1}{\beta+1 } \exp \left( -\frac{\tau_1 }{\lambda_1}(\beta+1) \right) \left(\frac{1}{\lambda_1^\beta}-\frac{1}{\lambda_2^\beta} \right)
    	\\
    	&+\exp\left(-\frac{\tau_2+h}{\lambda_2}(\beta+1) \right) \left(1 - \frac{1}{\lambda_2^\beta (\beta+1)} \right),
    	\\
    	h_2(\bm{\lambda})&=- \frac{\beta+1}{\beta N} \left\{ \frac{1}{\lambda_1^\beta}\sum_{i=1}^{n_1}\exp\left( -\frac{t_{i:N}}{\lambda_1}\beta \right) 
    	+\frac{1}{\lambda_2^\beta} \sum_{i=n_1+1}^{n_1+n_2} \exp\left(-\frac{t_{i:N}+h}{\lambda_2}\beta \right) \right.
    	\\
    	& + \left. \left(N-n_1-n_2 \right) \exp\left(- \frac{\tau_2+h}{\lambda_2}\beta\right) \right\}.
    \end{align*} 
\end{proposition}
\par
Note that the term $\frac{1}{\beta}$ has been removed, as it does not interfere in the minimization process.
As pointed out in \cite{jaenada2025}, the family of MDPDE can be extended at $\beta=0$, yielding the MLE defined.

To extrapolate results to normal operating conditions, a model relating the stress level, $x$, with the distribution parameter $\lambda$ is required. As is standard in the ALT literature, we assume a log-linear relationship between the stress level and the scale parameter of the underlying lifetime distribution:
\begin{equation}\label{eq:lambda_eq}
    \lambda_i = \exp(a_0 + a_1 x_i), \quad i=1, 2.
\end{equation}
This log-linear relationship implies that the logarithm of the scale parameter depends linearly on the applied stress level, since $\log(\lambda_i) = a_0 + a_1 x_i$. Under this formulation, the parameters to be estimated are no longer the scale parameters $\bm{\lambda} = (\lambda_1, \lambda_2)^T$ directly, but rather the regression coefficients $\bm{a} = (a_0, a_1)^T$. 
This formulation represents a direct reparameterization of the model, maintaining the same number of parameters (two unknown parameters in both cases). Consequently, once the estimates of $a_0$ and $a_1$ are obtained, they can be directly used for extrapolation to any normal operating stress condition $x_0$, allowing for the calculation of the corresponding scale parameter $\lambda_0 = \exp(a_0 + a_1 x_0)$.

\section{The Restricted MDPDE (RMDPDE)} \label{sec:sec3}

Hypothesis testing will be considered over restricted subspaces defined by linear functions of the parameters. This framework is highly advantageous as it allows us to formally test specific coefficient values and determine whether the stress factor has a statistically significant effect on the product's lifetime (e.g., by testing the null hypothesis $H_0: a_1 = 0$).
Given the following restriction for $\bm{a}=(a_0,a_1)^{\intercal}$:
\begin{align}
g(a_0,a_1) &= \bm{m}^{\intercal} 
\begin{pmatrix}
  a_{0} \\
  a_{1}
\end{pmatrix} - d = 0,
\label{eq:RMDPDE_eq}
\end{align}
where $\bm{m}=(m_0,m_1)^{\intercal} \in \mathbb{R}^2$ and $d \in \mathbb{R}$. RMDPDE is defined by
\begin{align}
\tilde{\bm{a}}^{\beta}&= \arg \min_{\bm{a} \in A_0} d_{\beta}\left(G_{T}(\cdot),F_{T}(\cdot|\bm{\lambda})\right),
\end{align}
where $A_0$ is the subspace of solutions of  \eqref{eq:RMDPDE_eq}.
Using the method of Lagrange multipliers and the functions used in \ref{prop:prop1},  the Lagrangian is given by
\begin{align*}
L(a_0,a_1,\tilde{\mu})&=h_{1}(a_0,a_1)+h_{2}(a_0,a_1)+\tilde{\mu}(m_0a_0+m_1a_1-d).
\end{align*}
and the first order conditions are
\begin{align}
\frac{\partial h_{1}(a_0,a_1)}{\partial a_0}+\frac{\partial h_{2}(a_0,a_1)}{\partial a_0}+\tilde{\mu}m_0&=0 \nonumber
\\
\frac{\partial h_{1}(a_0,a_1)}{\partial a_1}+\frac{\partial h_{2}(a_0,a_1)}{\partial a_1}+\tilde{\mu}m_1&=0 \nonumber
\\
m_0a_0+m_1a_1-d&=0.
\end{align}
In the case of the exponential lifetime we have,
\begin{align*}
	\frac{\partial h_{1}(a_0,a_1)}{\partial a_0}&=-\frac{\beta}{\lambda_1^\beta (\beta+1)}- \exp\left( -(\beta+1)\frac{\tau_1}{\lambda_1} \right) \left[ \frac{\tau_1}{\lambda_1} \left( \frac{1}{\lambda_1^\beta} - \frac{1}{\lambda_2^\beta} \right) - \frac{\beta}{\beta+1} \left( \frac{1}{\lambda_1^\beta} - \frac{1}{\lambda_2^\beta} \right) \right] \\
&\quad + \exp\left( -(\beta+1)\frac{\tau_2+h}{\lambda_2} \right) \left[ (\beta+1)\frac{\tau_2+h}{\lambda_2} \left( 1 - \frac{1}{\lambda_2^\beta(\beta+1)} \right) +\frac{\beta}{\lambda_2^\beta(\beta+1)} \right],
	\\
	\frac{\partial h_{2}(a_0,a_1)}{\partial a_0}&=- \frac{\beta+1}{N} \Bigg[
	\sum_{i=1}^{n_1} \frac{1}{\lambda_1^\beta} \left(\frac{t_{i:N}}{\lambda_1}-1\right) \exp\left(-\beta \frac{t_{i:N}}{ \lambda_1}\right)
	 \\
	&\quad + \sum_{i=n_1+1}^{n_1+n_2} \frac{1}{\lambda_2^\beta} \left(\frac{t_{i:N}+h}{\lambda_2}-1\right) \exp\left(-\beta \frac{t_{i:N}+h}{ \lambda_2}\right)
	 \\
	&\left. \quad + (N-n_1-n_2) \frac{\tau_2+h}{\lambda_2} \exp\left(-\beta \frac{\tau_2+h}{\lambda_2}\right)\right]
,
	\\
\frac{\partial h_1(a_0, a_1)}{\partial a_1} &=  -\frac{\beta x_1}{\lambda_1^{\beta}(\beta+1)} 
	\\
	& - \left[ \frac{\tau_1 x_1}{\lambda_1} \left( \frac{1}{\lambda_1^{\beta}} - \frac{1}{\lambda_2^{\beta}} \right) - \frac{\beta}{\beta+1} \left( \frac{x_1}{\lambda_1^{\beta}}-\frac{x_2}{\lambda_2^{\beta}}  \right) \right] \exp\left( -\frac{\tau_1}{\lambda_1}(\beta+1) \right)
	 \\
	& + \left[ (\beta+1) \left( \frac{(\tau_2-\tau_1)x_2+\frac{\lambda_2}{\lambda_1}\tau_1 x_1}{\lambda_2}  \right) \left( 1 - \frac{1}{\lambda_2^{\beta}(\beta+1)} \right) + \frac{\beta x_2}{\lambda_2^{\beta}(\beta+1)} \right] \exp\left( -\frac{\tau_2+h}{\lambda_2}(\beta+1) \right),
	\\
	\frac{\partial h_{2}(a_0,a_1)}{\partial a_1}&=-\frac{\beta+1}{ N} \left\{ \frac{x_1}{\lambda_1^{\beta}}\sum_{i=1}^{n_1}\left(\frac{t_{i:N}}{\lambda_1}-1\right)\exp\left(-\frac{t_{i:N} \beta}{\lambda_1}\right)  \right.
	\\
	&+\frac{1}{\lambda_2^{\beta}}\sum_{i=n_1+1}^{n_1+n_2}\left(\frac{(t_{i:N}-\tau_1)x_2 + \frac{\lambda_2}{\lambda_1}\tau_1 x_1}{\lambda_2}-x_2\right)\exp\left(-\frac{t_{i:N} +h}{\lambda_2}\beta\right) 
	\\
	&\left. +(N-n_1-n_2)\frac{(\tau_2-\tau_1)x_2+\frac{\lambda_2}{\lambda_1}\tau_1x_1}{\lambda_2} \exp\left(-\frac{\tau_2+h}{\lambda_2}\beta\right) \right\}.
\end{align*} 
For more details see \citet{jaenada2025}.
So the RMDPDE $\tilde{\bm{a}}$ is the vector that must satisfy those equations.
\begin{theorem}\label{th:theo_1}
Let $\bm{a}^*=(a^*_0,a^*_1)^{\intercal}$ be the true value of the parameters and assume that $g(a_0^*,a_1^*)=0$
 with $g(\cdot)$ defined in  \eqref{eq:RMDPDE_eq}.The asymptotic distribution of the RMDPDE for the SSALTs model under exponential lifetime, $\tilde{\bm{a}}^{\beta}$, obtained under constraint $g(\bm{a})=0$ is given by:
 \begin{align*}
 \sqrt{N}\left(\tilde{\bm{a}}^\beta - \bm{a}^*\right) \xrightarrow[N \to \infty]{\mathcal{L}} \mathcal{N}\left(\bm{0},\bm{\Sigma}_{\beta}(\bm{a}^*)\right),
 \end{align*}
 where
 \begin{align*}
 \bm{\Sigma}_{\beta}(\bm{a}^*)&=\bm{P}_\beta(\bm{a}^*)\bm{K}_\beta(\bm{a}^*)\bm{P}_\beta(\bm{a}^*)
 \\
\bm{P}_\beta(\bm{a}^*)&=\bm{J}_\beta(\bm{a}^*)^{-1}-\bm{Q}_\beta(\bm{a}^*)\bm{m}^{\intercal}\bm{J}_\beta(\bm{a}^*)^{-1}
  \\
   \bm{Q}_{\beta}(\bm{a}^*)&=\bm{J}_\beta(\bm{a}^*)^{-1} \bm{m}\left(\bm{m}^{\intercal}\bm{J}_\beta(\bm{a}^*)^{-1}\bm{m}\right)^{-1},
 \end{align*}
 where $\bm{J}_{\beta}(\bm{a}^*)$ a matrix whose components are
 \begin{align*}
 J_{ii}^{\beta}(a_{i}^*)&=\int_0^{\tau_1}\left(\frac{\partial \log\left(f_1(t|\bm{\lambda})\right)}{\partial a_i}\right)^{2}f_1(t|\bm{\lambda})^{\beta+1}dt
 \\
 &+\int_{\tau_1}^{\tau_2}\left(\frac{\partial \log\left(f_2(t+h|\bm{\lambda})\right)}{\partial a_i}\right)^{2}f_2(t+h|\bm{\lambda})^{\beta+1}dt
 \\
 &+	\left(\frac{\partial \log\left(1-F_2(\tau_2+h|\bm{\lambda})\right)}{\partial a_i}\right)^{2}\left(1-F_2(\tau_2+h|\bm{\lambda})\right)^{\beta+1}
 \\
& \text{for i=0,1},
 \end{align*}
  \begin{align*}
 J_{ij}^{\beta}(a_{i}^*,a_{j}^*)&=\int_0^{\tau_1}\left(\frac{\partial \log\left(f_1(t|\lambda_1)\right)}{\partial a_i}\right)\left(\frac{\partial \log\left(f_1(t|\bm{\lambda})\right)}{\partial a_j}\right)f_1(t|\bm{\lambda})^{\beta+1}dt
 \\
 &+\int_{\tau_1}^{\tau_2}\left(\frac{\partial \log\left(f_2(t+h|\bm{\lambda})\right)}{\partial a_i}\right)\left(\frac{\partial \log\left(f_2(t+h|\bm{\lambda})\right)}{\partial a_j}\right)f_2(t+h|\bm{\lambda})^{\beta+1}dt
 \\
 &+	\left(\frac{\partial \log\left(1-F_2(\tau_2+h|\bm{\lambda})\right)}{\partial a_i}\right)\left(\frac{\partial \log\left(1-F_2(\tau_2+h|\bm{\lambda})\right)}{\partial a_j}\right)\left(1-F_2(\tau_2+h|\bm{\lambda})\right)^{\beta+1}
 \\
& \text{for } i,j=0,1; \, i \neq j,
 \end{align*}
 and
 \begin{align*}
 \bm{K}_{\beta}\left(\bm{a}^*\right)&= \bm{J}_{2\beta}\left(\bm{a}^*\right)-\bm{\zeta}_{\beta}\left(\bm{a}^*\right)^{\intercal}\bm{\zeta}_{\beta}\left(\bm{a}^*\right),
 \end{align*}
 with $\bm{\zeta}_\beta(\bm{a}^*)$ a vector whose entries are:
 \begin{align*}
 \zeta_i^{\beta}(\bm{a}^*)&=\int_0^{\tau_1}\left(\frac{\partial \log\left(f_1(t|\lambda_1)\right)}{\partial a_i}\right)f_1(t|\lambda_1)^{\beta+1}dt
 \\
 &+\int_{\tau_1}^{\tau_2}\left(\frac{\partial \log\left(f_2(t+h|\bm{\lambda})\right)}{\partial a_i}\right)f_2(t+h|\bm{\lambda})^{\beta+1}dt
 \\
 &+	\left(\frac{\partial \log\left(1-F_2(\tau_2+h|\bm{\lambda})\right)}{\partial a_i}\right)\left(1-F_2(\tau_2+h|\bm{\lambda})\right)^{\beta+1}
 \\
& \text{for i=0,1},
 \end{align*}
 \end{theorem}
 where $\lambda_i$ depends on $a_0$ and $a_1$ through the relationship defined in \eqref{eq:lambda_eq}.
\begin{proof}
See Appendix.
\end{proof}
The  RMDPDE, $\tilde{\bm{a}}^\beta$, derived in the previous section is strictly necessary to construct the Rao-type test statistics. To proceed with the formulation of this test, let us first formalize our testing framework. We consider a linear null hypothesis on the parameter $\bm{a}$ of the form
\begin{align*}
H_{0}:\bm{m}^{\intercal}\bm{a}&=d,
\end{align*}
where $\bm{m}=(m_0,m_1)^{\intercal} \in \mathbb{R}^2$ and $d \in \mathbb{R}$. In particular, the linear null hypothesis with $\bm{m}^{\intercal}=(0,1)$ and $d=0$ would test if the stress level affects the lifetime of the devices.

\section{Z-type statistics}\label{sec:sec4}

In this section, we introduce the Z-type test statistics based on the unrestricted MDPDE. To ensure consistency throughout our testing procedures, we will maintain the same structure for the parametric space, focusing exclusively on linear constraints as defined for the null hypothesis above.

\begin{definition} \label{def:def_z}
The Z-type statistic based on the MDPDE $\hat{\bm{a}}^{\beta}$, for testing the null hypothesis is given by:
\begin{align}
Z_{N}(\hat{\bm{a}}^{\beta})=\sqrt{N}\left(\bm{m}^{\intercal}\bm{J}_{\beta}\left(\hat{\bm{a}}^{\beta}\right)^{-1}\bm{K}_{\beta}\left(\hat{\bm{a}}^{\beta}\right)\bm{J}_{\beta}\left(\hat{\bm{a}}^{\beta}\right)^{-1}\bm{m}\right)^{-\frac{1}{2}}\left(\bm{m}^{\intercal}\hat{\bm{a}}^{\beta}-d\right).
\label{eq:z_type}
\end{align}
 \end{definition}
 The asymptotic distribution of the statistic is given by the following theorem.
 \begin{theorem}\label{th:theo_2}
 The asymptotic distribution of the Z-statistic under the null hypothesis is a standard Normal distribution.
 \end{theorem}

\begin{proof}
See Appendix.
\end{proof}
Based on Theorem \ref{th:theo_2} , for any $\beta \geq0$ and $\bm{m} \in \mathbb{R}^2$, the critical region of significance level $\alpha$ for a hypothesis test with a linear null hypothesis, is given by:
\begin{align*}
R_{\alpha}=\left\{\left(t_{1:N};...;t_{n_1:N};...;t_{n_1+n_2:N};N-n_1-n_2\right)\left| \left|Z_{N}\left(\hat{\bm{a}}^{\beta}\right)\right| > Z_{\alpha/2} \right. \right\},
\end{align*}
where $t_{1:N};...;t_{n_1:N};...;t_{n_1+n_2:N}$ are the times of failure in the continuous interval and $N-n_1-n_2$ are the number of devices that survive at the end of the test.
\begin{remark}
We can generalize the null hypothesis to:
\begin{align*}
H_0:\bm{M}^{\intercal}\bm{a}=\bm{d},
\end{align*}
where $\bm{M}$ is a $r \times 2$ matrix $(r \leq 2)$ and d being an r-dimensional vector. Then the corresponding test statistic can be defined
\begin{align*}
Z_{N}\left(\hat{\bm{a}}^{\beta}\right)=N\left(\bm{M}^{\intercal}\hat{\bm{a}}^{\beta}-d\right)^{\intercal}\left(\bm{M}^{\intercal}\bm{J}_{\beta}\left(\hat{\bm{a}}^{\beta}\right)^{-1}\bm{K}_{\beta}\left(\hat{\bm{a}}^{\beta}\right)\bm{J}_{\beta}\left(\hat{\bm{a}}^{\beta}\right)^{-1}\bm{M}\right)^{-1}\left(\bm{M}^{\intercal}\hat{\bm{a}}^{\beta}-d\right).
\end{align*}
The matrix $\bm{M}^{\intercal}\bm{J}_{\beta}\left(\hat{\bm{a}}^{\beta}\right)^{-1}\bm{K}_{\beta}\left(\hat{\bm{a}}^{\beta}\right)\bm{J}_{\beta}\left(\hat{\bm{a}}^{\beta}\right)^{-1}\bm{M}$ is symmetric so the statistic is well defined.
With a similar proof of the theorem above one can show that under the generalized null hypothesis:
\begin{align*}
\sqrt{N}\left(\bm{M}^{\intercal}\bm{J}_{\beta}\left(\hat{\bm{a}}^{\beta}\right)^{-1}\bm{K}_{\beta}\left(\hat{\bm{a}}^{\beta}\right)\bm{J}_{\beta}\left(\hat{\bm{a}}^{\beta}\right)^{-1}\bm{M}\right)^{-\frac{1}{2}}\left(\bm{M}^{\intercal}\hat{\bm{a}}^{\beta}-\bm{d}\right) \xrightarrow[N \to \infty]{\mathcal{L}} \mathcal{N}\left(\bm{0},\bm{I}_{r\times r}\right).
\end{align*}
So
\begin{align*}
Z_{N}\left(\hat{\bm{a}}^{\beta}\right)\xrightarrow[N \to \infty]{\mathcal{L}} \chi_{r}^{2}.
\end{align*}
\end{remark}

\begin{theorem}\label{th:theo_3}
Let $\bm{a}^* \in A$ be the true parameter of $\bm{a}$ with $\bm{m}^{\intercal}\bm{a}^* \neq d$.  Then, the approximate power function of the test statistic is given by:
\begin{align*}
B_{N}\left(\bm{a}^*\right) &\approx \Phi\left(-z_{\alpha/2}-\sqrt{N}\left(\bm{m}^{\intercal}\bm{J}_{\beta}\left(\hat{\bm{a}}^{\beta}\right)^{-1}\bm{K}_{\beta}\left(\hat{\bm{a}}^{\beta}\right)\bm{J}_{\beta}\left(\hat{\bm{a}}^{\beta}\right)^{-1}\bm{m}\right)^{-\frac{1}{2}}\left(\bm{m}^{\intercal}\bm{a}^*-d\right)\right)
\\
&+1-\Phi\left(z_{\alpha/2}-\sqrt{N}\left(\bm{m}^{\intercal}\bm{J}_{\beta}\left(\hat{\bm{a}}^{\beta}\right)^{-1}\bm{K}_{\beta}\left(\hat{\bm{a}}^{\beta}\right)\bm{J}_{\beta}\left(\hat{\bm{a}}^{\beta}\right)^{-1}\bm{m}\right)^{-\frac{1}{2}}\left(\bm{m}^{\intercal}\bm{a}^{*}-d\right)\right).
\end{align*}
\end{theorem}

\begin{proof}
See Appendix.
\end{proof}

\subsection{Testing Contiguous Alternative Hypothesis}
Let $\bm{a}_L \in A \setminus A_0$ be an alternative and take $\bm{a}^*$ the closest element to the boundary of $A_0$ with respect to the Euclidean distance. Consider the alternative as:
\begin{align}\label{eq:Null_test}
H_{1,L}: \bm{a}=\bm{a}_L,
\end{align}
with $\bm{a}_L=\bm{a}^*+\frac{1}{\sqrt{N}}\ell$ for a fixed vector $\bm{\ell} \in \mathbb{R}^2$ and defining $\bm{\ell}^*=\bm{m}^{\intercal} \bm{\ell}$ we have:
\begin{align*}
\bm{m}^{\intercal}\bm{a}_L - d =\bm{m}^{\intercal}\left(\bm{a}_L-\bm{a}^*\right)=\bm{m}^{\intercal}\frac{\bm{\ell}}{\sqrt{N}}=\frac{\bm{\ell}^*}{\sqrt{N}}.
\end{align*}
So the contiguous hypothesis is equivalent to:
\begin{align*}
g(\bm{a}_L)=\frac{\bm{\ell}^*}{\sqrt{N}}.
\end{align*}
\begin{theorem}\label{th:theo_4}
The asymptotic distribution of the Z-type statistic under the contiguous hypothesis is a normal distribution with mean $\left(\bm{m}^{\intercal}\bm{J}_{\beta}\left(\hat{\bm{a}}^{\beta}\right)^{-1}\bm{K}_{\beta}\left(\hat{\bm{a}}^{\beta}\right)\bm{J}_{\beta}\left(\hat{\bm{a}}^{\beta}\right)^{-1}\bm{m}\right)^{-\frac{1}{2}} \bm{m}^{\intercal}\bm{\ell}$ and unit variance.
\end{theorem}
\begin{proof}
    See Appendix.
\end{proof}
The power function of the test statistic can be obtained as
\begin{align*}
B_{N}\left(\bm{a}_L\right)&=P\left(\left|Z_{N}\left(\hat{\bm{a}}^{\beta} \right) \right| \ge z_{\alpha/2} \left| \bm{a}=\bm{a}_L \right. \right)
\\
&=P\left(Z_{N}\left(\hat{\bm{a}}^{\beta} \right)  \le -z_{\alpha/2} \left| \bm{a}=\bm{a}_L \right. \right) +P\left(Z_{N}\left(\hat{\bm{a}}^{\beta} \right)  \geq z_{\alpha/2} \left| \bm{a}=\bm{a}_L \right. \right)
\\
&=\Phi\left(-z_{\alpha/2}-\sqrt{\frac{N}{\bm{m}^{\intercal}\bm{J}_{\beta}\left(\hat{\bm{a}}_L\right)^{-1}\bm{K}_{\beta}\left(\hat{\bm{a}}_L\right)\bm{J}_{\beta}\left(\hat{\bm{a}}_L\right)^{-1}\bm{m}}} \bm{m}^{\intercal}\bm{\ell}\right)
\\
&+1-\Phi\left(z_{\alpha/2}-\sqrt{\frac{N}{\bm{m}^{\intercal}\bm{J}_{\beta}\left(\hat{\bm{a}}_L\right)^{-1}\bm{K}_{\beta}\left(\hat{\bm{a}}_L\right)\bm{J}_{\beta}\left(\hat{\bm{a}}_L\right)^{-1}\bm{m}}} \bm{m}^{\intercal}\bm{\ell}\right).
\end{align*}
It is clear that $\lim_{N \to \infty}B_{N}\left(\bm{a}_L\right)=1$. Therefore, the test is consistent in Fraser's sense.

\section{Rao-type Test Statistic}\label{sec:sec6}
Rao's test uses the restricted MLE and may outperform the Wald test. Basu developed parametric tests for composite hypotheses based on RMDPDEs.
Let us consider $\tilde{\bm{a}}^{\beta}$ be the restricted MDPDE, where the parameter space is constrained by the null hypothesis:
\begin{align*}
\bm{A}_0= \left\{\bm{a} \mid \bm{m}^{\intercal}\bm{a}=d \right\} ,
\end{align*}
and $ \hat{\bm{a}}^{\beta}$ denotes the MDPDE for $\bm{a}$ over the whole parameter space. Let us consider the score function associated with the DPD loss for this model:
\begin{align}\label{eq:score}
U_{\beta,N}\left(\bm{a}\right)=\frac{\partial h_{1}\left(a_0,a_1\right)}{\partial \bm{a}} +\frac{\partial h_{2}\left(a_0,a_1\right)}{\partial \bm{a}} .
\end{align}
Remember that, based on the minimization problem, the MDPDE satisfies:
\begin{align*}
U_{\beta,N}\left(\hat{\bm{a}}^{\beta}\right)=\bm{0}.
\end{align*}
\begin{definition}\label{def:rao_test_def}
The Rao-type statistic, based on the RMDPDE $\tilde{\bm{a}}^{\beta}$, for testing the null hypothesis given in \eqref{eq:Null_test} is given by
\begin{align} \label{eq:rao_test_def}   
R_{\beta,N}\left(\tilde{\bm{a}}^{\beta}\right)=N\, U_{\beta,N}\left(\tilde{\bm{a}}^{\beta}\right)^{\intercal}Q_{\beta}\left(\tilde{\bm{a}}^{\beta}\right)\left[ Q_{\beta}\left(\tilde{\bm{a}}^{\beta}\right)^{\intercal}K_{\beta}\left(\tilde{\bm{a}}^{\beta}\right)Q_{\beta}\left(\tilde{\bm{a}}^{\beta}\right)\right]^{-1}Q_{\beta}\left(\tilde{\bm{a}}^{\beta}\right)^{\intercal}U_{\beta}\left(\tilde{\bm{a}}^{\beta}\right).
\end{align}
\end{definition}
\begin{theorem}\label{th:theo_5}
The asymptotic distribution of the score $U_{\beta,N}\left(\tilde{\bm{a}}^{\beta}\right)$, under correct specification of the model, for the step-stress ALT with continuous monitoring and exponential lifetimes is given by:
\begin{align*}
\sqrt{N}\,U_{\beta,N}\left(\tilde{\bm{a}}^{\beta}\right) \xrightarrow[N \to \infty]{\mathcal{L}} \mathcal{N}\left(\bm{0},\bm{K}\left(\bm{a}^*\right)\right).
\end{align*}
\end{theorem}
\begin{proof}
See Appendix.
\end{proof}
\begin{theorem}\label{th:theo_6}
The asymptotic distribution of the Rao-type test statistic under the linear null hypothesis is a $\chi^2_1$.
\end{theorem}

\begin{proof}
See Appendix.
\end{proof}

  \section{Simulation Study}\label{sec:sec7}

  This section investigates the finite-sample performance of the RMDPDE, the (Z)-type tests, and the Rao-type tests for the step-stress ALT model under exponential lifetimes. To this end, a simple step-stress accelerated life test (SSALT) under Type-I censoring is simulated according to the following design.

A total of $n=360$ units, whose lifetimes follow an exponential distribution, are subjected to the step-stress test. Initially, all units are tested at the stress level $x_1=1$. At the stress-change time $\tau_1=9$, the stress level is increased to $x_2=2$. The experiment is terminated at time $\tau_2=24.55$, and the nominal operating stress is set to $x_0=0.5$. The true parameter vector is $\boldsymbol{a} = (a_0, a_1)^\intercal = (3.65, -1.15)^\intercal$. 
Under the log-linear relationship defined in \eqref{eq:lambda_eq}, this true parameter vector determines the scale parameters (which coincide with the mean lifetimes for the exponential distribution) for each stress condition. Specifically, the expected mean lifetimes are $\lambda_0 \approx 21.65$ under the nominal stress $x_0$, $\lambda_1 \approx 12.18$ under the initial stress $x_1$, and $\lambda_2 \approx 3.86$ under the accelerated stress $x_2$. Notice how the mean lifetime strictly decreases as the applied stress increases, a physical requirement reflected by the negative sign of the parameter $a_1$.
This setup dictates the expected behavior of the failure process throughout the experiment. By the stress-change time $\tau_1=9$, the expected probability of failure is $F^*(\tau_1) \approx 0.522$, indicating that roughly $52.2\%$ of the units are expected to fail during the first stage of the test. By the end of the experiment at $\tau_2=24.55$, the cumulative probability of failure reaches $F^*(\tau_2) \approx 0.991$. Consequently, only about $0.9\%$ of the sample is expected to be right-censored, ensuring a high proportion of exact failure times to robustly estimate the parameters.

To evaluate the robustness of the estimators, it is necessary to specify a mechanism to generate outliers and define the contamination scheme. In our framework, outliers are conceptualized as an unexpectedly high proportion of observations falling within a low-density region of the target theoretical distribution. Under the assumed exponential step-stress model, the probability density decreases rapidly over time; thus, observing a cluster of failures in the far right tail is highly improbable under clean conditions. By artificially generating a small mass of observations in these low-probability regions, we simulate a realistic contamination scenario.
In contrast to probability-matching approaches, contamination is introduced directly at the parameter level. Specifically, for a given contamination intensity $\epsilon$, the parameter is perturbed multiplicatively as
\[
\tilde{a} =
\begin{cases}
(1+\epsilon)a_0, & \text{if $a_0$ is contaminated}\\
(1-\epsilon)a_1 & \text{if $a_1$ is contaminated}. 
\end{cases}
\]
This transformation is applied either to $a_0$ or to $a_1$, while keeping the other parameter fixed. 

We increase $a_0$ via $(1+\epsilon)$ and decrease the magnitude of $a_1$ via $(1-\epsilon)$ (recalling that $a_1 < 0$, which makes this shift an increase in its algebraic value). Under both scenarios, this directional change directly results in larger scale parameters $\lambda_2$, thereby increasing the expected mean lifetime of the units. Consequently, the probability mass is shifted towards much larger failure times, resulting in a lower hazard rate during the active testing intervals. This shift effectively forces a portion of the simulated sample to behave as outliers by clustering failures in regions that would otherwise exhibit extremely low theoretical density under the clean model.

Outliers are generated within a fixed time interval $[T_1, T_2] \subset (\tau_1, \tau_2)$, which in this study is set to $(T_1, T_2) = (21.5, 24.5)$. Under the nominal model, the probability of a unit failing within this specific interval is approximately $0.01$. This extremely low probability mathematically guarantees that the theoretical probability density in this region is very low. Consequently, introducing contaminated observations here forces a cluster of failures in a low-density zone, providing a highly reliable scheme to evaluate the robustness of the estimators.

The proportion of outliers is not imposed directly, but instead has been derived from the contaminated model in each case. More precisely, given the contaminated parameters $(\tilde{a}_0, \tilde{a}_1)$, the outlier proportion $\epsilon_{\text{eff}}$ is computed as the probability mass assigned by the contaminated distribution into the contaminated interval $[T_1, T_2]$. That is,
\[
\text{proportion}_\text{outlier}= P(T_1 < T < T_2|\tilde{\bm{a}}),
\]
which depends implicitly on the contamination level and the stress configuration.

The proportion $\epsilon_{\text{eff}}$ of contaminated observations is generated from a shifted and truncated exponential distribution over $[T_1, T_2]$ with the same contaminated parameters as used to determine the contamination proportion, while the remaining observations follow the original (non-contaminated) step-stress exponential model. In this way, the contamination is twofold: on the one hand, a higher proportion of failures is introduced within a low-density interval; on the other hand, these outlying observations follow a different underlying distribution. This construction ensures that the presence of outliers is driven by a structurally different failure mechanism rather than an artificial adjustment of probabilities. Furthermore, retaining the shifted exponential nature for the outliers within this interval realistically captures the physical wear-out process of defective units, which are more prone to fail immediately upon entering the stress phase (due to higher initial density) rather than uniformly across the interval.
Thus, this approach reflects realistic scenarios in accelerated life testing, where a subset of units may be affected by manufacturing defects or experimental anomalies, resulting in lifetimes governed by a different parameter configuration.

Since this work ultimately deals with hypothesis testing, we naturally contemplate two distinct scenarios regarding the underlying data-generating process. This design ensures that when the statistical tests and the restricted estimators are formally defined, we can evaluate their performance under both a true null hypothesis (correct specification) and a false null hypothesis (misspecification). Specifically, we establish the hypothesis test for the stress parameter by defining the null hypothesis as 
\begin{equation}\label{eq:hypothesis_constraint}
H_0 : \bm{m}^T \bm{a} = -1.15,
\end{equation}
where $\bm{m}^T = (0, 1)$ and $\bm{a} = (a_0, a_1)^T$. This restricts the general parameter space $\Theta$ to the null subspace $\Theta_0 = \{ \boldsymbol{\theta} \in \Theta \mid a_1 = -1.15 \}$. Under this framework, we analyze two distinct scenarios depending on the true value of the stress parameter: a correctly specified model ($a_1 = -1.15$) and a misspecified model ($a_1 = -1.4$). Because the underlying lifetime distribution is directly governed by $a_1$, the theoretical probability density functions differ between these two setups. Consequently, for a fixed contamination parameter $\varepsilon$, the resulting proportion of simulated outliers varies from one model to the other. Tables \ref{tab:outlier_a0} and \ref{tab:outlier_a1} detail these empirical outlier proportions across the different contamination levels. As expected, the volume of outliers scales with $\varepsilon$, ultimately representing between 3\% and 5\% of the sample at the highest contamination levels under both scenarios.
\begin{table}[H]
\centering
\caption{Proportion of outliers (\%) for different contamination levels ($\varepsilon$) when introducing $a_0$ outliers.}
\label{tab:outlier_a0}
\begin{tabular}{lcc}
\toprule
& \textbf{Correctly Specified Model (\%)} & \textbf{Misspecified Model (\%)} \\
\cmidrule(lr){2-2} \cmidrule(lr){3-3}
$\varepsilon$ & ($a_0$ outlier) & ($a_0$ outlier) \\
\midrule
0.00 & 0.00 & 0.00 \\
0.06 & 1.66 & 0.37 \\
0.12 & 2.62 & 0.86 \\
0.18 & 3.63 & 1.63 \\
0.24 & 4.51 & 2.61 \\
0.30 & 5.15 & 3.66 \\
\bottomrule
\end{tabular}
\end{table}
\begin{table}[H]
\centering
\caption{Proportion of outliers (\%) for different contamination levels ($\varepsilon$) when introducing $a_1$ outliers.}
\label{tab:outlier_a1}
\begin{tabular}{lcc}
\toprule
& \textbf{Correctly Specified Model (\%)} & \textbf{Misspecified Model (\%)} \\
\cmidrule(lr){2-2} \cmidrule(lr){3-3}
$\varepsilon$ & ($a_1$ outlier) & ($a_1$ outlier) \\
\midrule
0.000 & 0.00 & 0.00 \\
0.091 & 1.29 & 0.39 \\
0.182 & 1.75 & 0.92 \\
0.273 & 2.24 & 1.71 \\
0.364 & 2.75 & 2.63 \\
0.455 & 3.24 & 3.51 \\
\bottomrule
\end{tabular}
\end{table}

\subsection{Performance of the RMDPDE}

To evaluate the performance and robustness of the RMDPDE under linear constraints, we formally define the restricted estimation framework using the null hypothesis defined in \eqref{eq:hypothesis_constraint}. We can use the null space defined by that null hypothesis and analyze the behavior of the restricted estimator under different data-generating conditions depending on whether the imposed constraint is true or false:
For the correct specified model, the simulated data are generated using the true parameter value $a_1 = -1.15$. In this case, the true parameter vector belongs to the subspace defined by the null hypothesis, meaning that the restriction forced on the RMDPDE is correct. and for the misspecified model the simulated data are generated using the true parameter value $a_1 = -1.4$. Under this setup, the true parameter does not satisfy the constraint ($a_1 \notin H_0$), representing a scenario where we evaluate the RMDPDE under an incorrect model specification.
This dual analysis is key to understanding the trade-off inherent to the RMDPDE: it allows us to quantify the efficiency gains of using the restricted estimator when our prior constraints are accurate, as well as the potential bias introduced when the enforced hypothesis does not hold.

To evaluate the effects of outliers, we obtained the root of the mean square error (RMSE) of the RMDPDEs under different proportions of contamination $\epsilon$,  and tuning parameters values $\beta=0 (MLE), 0.2, ..., 1$.
We obtain the RMSE  by generating $R=1000$ samples, and for each one, obtaining the RMDPDE (note that in this scenario, the estimation is only for $\hat{a}_0$, because $\hat{a}_1=-1.15, -1.4$ for each simulation) So the RMDPDE for each tuning parameter and proportion of contamination is:
\begin{align*}
RMSE^{\beta,\epsilon}=\sqrt{\sum_{i=1}^{1000}\frac{\left(\hat{a}^{\beta,\epsilon}_{i}-3.65\right)^{2}}{1000}}.
\end{align*}
Figure \ref{fig:rmse_comparison_true} illustrates how the RMSE is affected by different proportions of contamination under the true assumption. The left plot displays the RMSE when the affected parameter is $a_0$, while the right plot shows the results for $a_1$. It can be observed that in the presence of outliers, the RMSE values are quite similar; in fact, with a sample size of 1,000 and incorporating the information of the true constraint ($a_1 = -1.15$), the RMSE remains very low.

In principle, in the absence of outliers, the MLE is the best estimator. However, as the proportion of outliers increases, the RMSE of all estimators rises. As expected, the presence of outliers negatively impacts the estimates. Notably, the RMSE does not vary significantly as the parameter $\beta$ increases. 

In summary, while the MLE performs better without outliers, it is much more sensitive to their presence than the RMDPDE estimators. This highlights the trade-off between efficiency and robustness: the RMDPDE is more robust, evidenced by the fact that with only a 2\% outlier proportion, the MLE already exhibits one of the highest RMSE values.

\begin{figure}[H]
    \centering
    \includegraphics[width=0.9\textwidth]{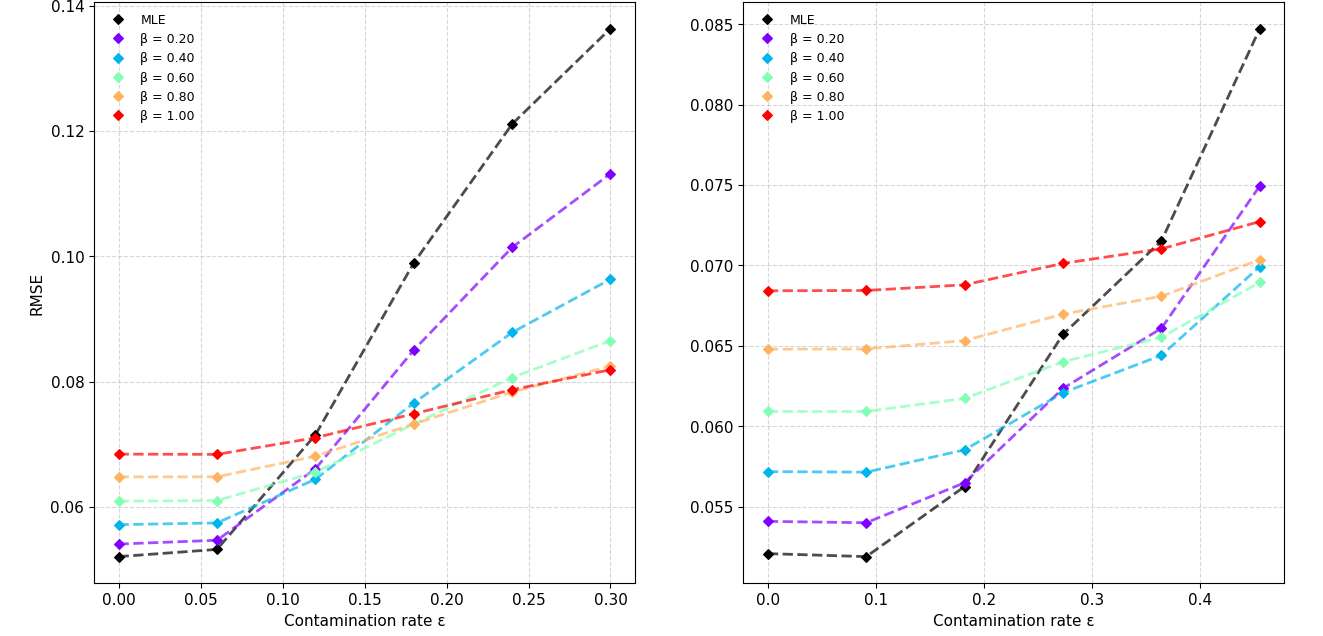}
    \caption{RMSE of the RMDPDE under the true model given contamination in $a_0$ (left) and in $a_1$ (right)}
    \label{fig:rmse_comparison_true}
\end{figure}

It is also of interest to measure the gain or loss in efficiency when using the RMDPDE instead of the RMLE. So, for each value of $\beta$, the relative efficiency can be obtained as follows:

\begin{equation*}
\rho(\hat{\bm{a}}^{\beta}) = \frac{RMSE^{\beta,\epsilon}}{RMSE^{0,\epsilon}}-1.
\end{equation*}
This metric compares the RMSE of the RMDPDE respect to the case of $\beta=0$, this is, the MLE. If $\rho(\hat{\theta}^{\beta})<0$, the RMDPDE has a lower RMSE than the MLE, meaning it is more efficient. It is clear that $\rho(\hat{\theta}^{MLE})=0$.
Figure \ref{fig:efficiency_comparison_true} illustrates the evolution of efficiency relative to the MLE. That is, although the reference value $\rho(\hat{\theta}^{MLE})$ remains constant at 0, it has already been established that the RMSE increases. It can be observed that as the proportion of outliers grows, the RMDPDE improves the MLE's RMSE by up to 20\%.
\begin{figure}[H]
    \centering
    \includegraphics[width=0.9\textwidth]{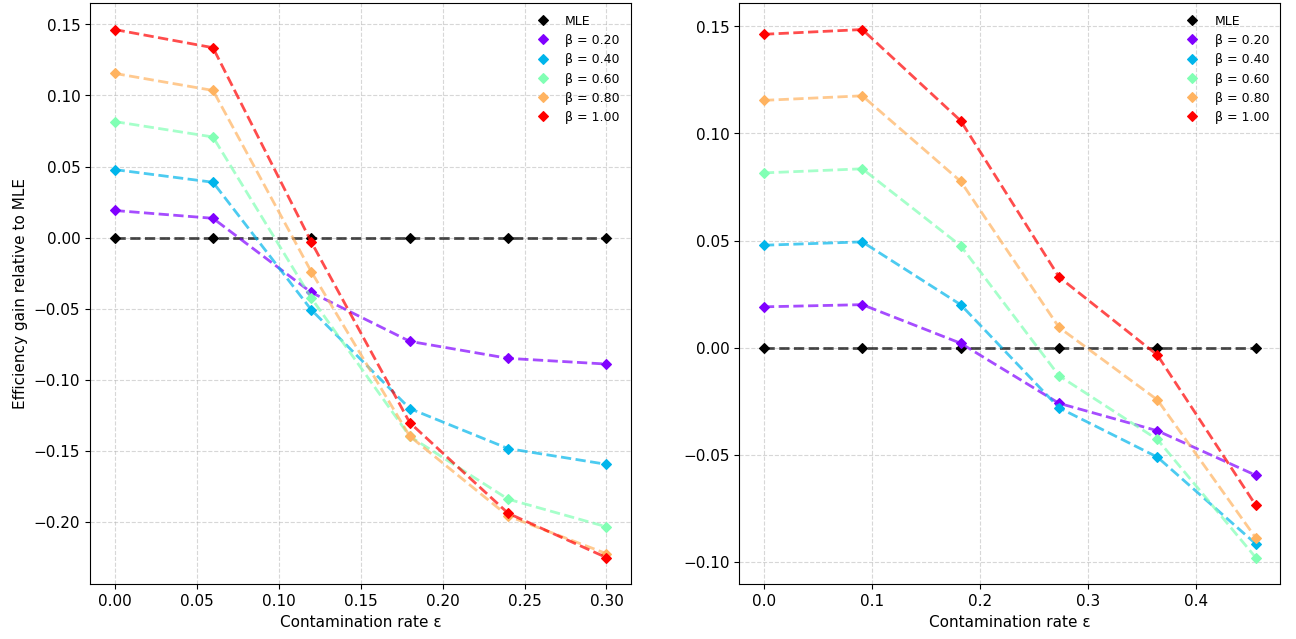}
    \caption{Efficiency gain of the RMDPDE under the true model given contamination in $a_0$ (left) and in $a_1$ (right)}
    \label{fig:efficiency_comparison_true}
\end{figure}
The RMSE was also calculated under a misspecified model ($a_1 = -1.4$ and $d = -1.15$), as shown in Figure~\ref{fig:rmse_comparison_false}. In this case, it can be observed that the RMSE for the MLE decreases in the presence of outliers. Although this might appear beneficial at first glance, it is important to consider that the model is misspecified. Consequently, if one intends to perform a hypothesis test, for instance, being closer to the true parameter values increases the probability of failing to reject the null hypothesis, thereby negatively affecting the decision-making process.
\begin{figure}[H]
    \centering
    \includegraphics[width=0.9\textwidth]{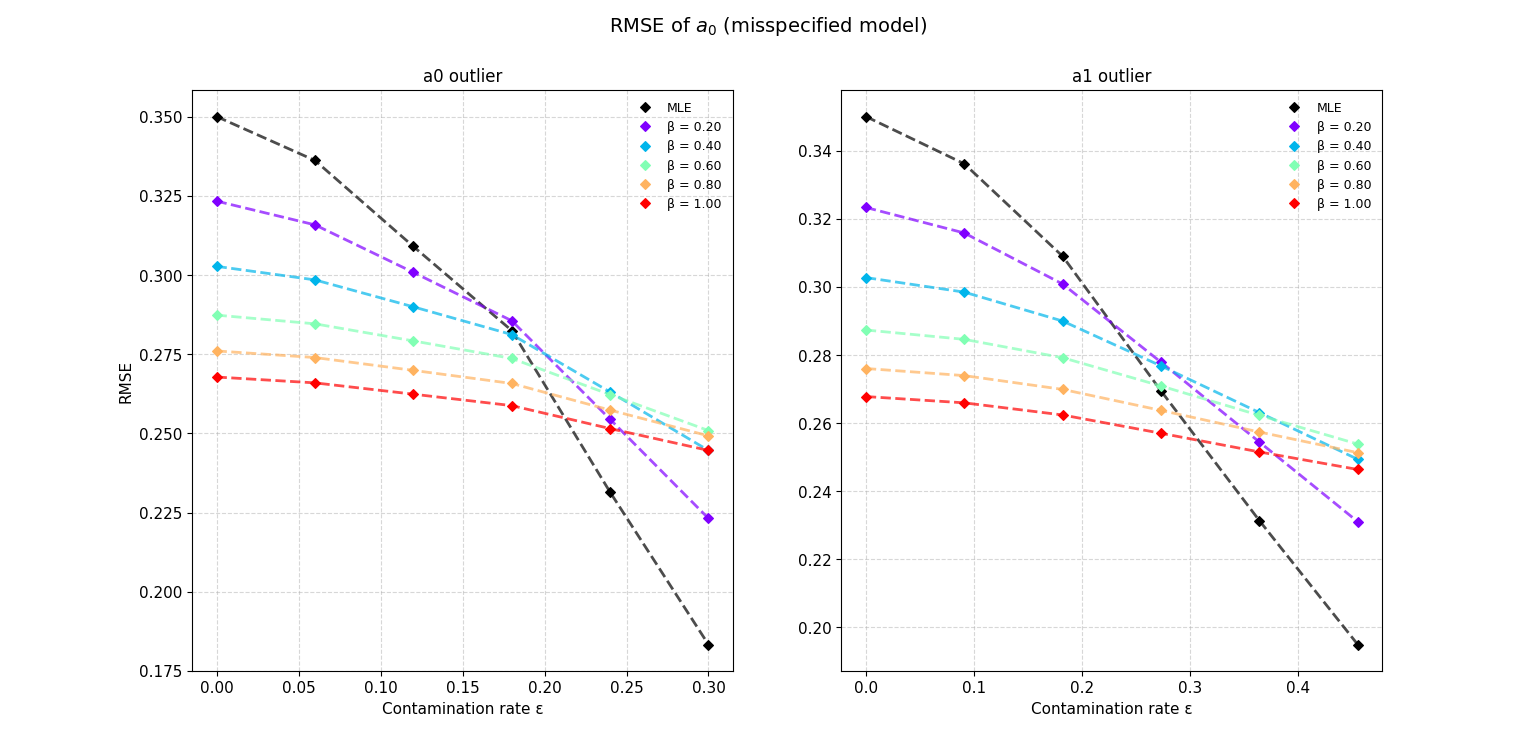}
    \caption{RMSE of the RMDPDE under the false model given contamination in $a_0$ (left) and in $a_1$ (right)}
    \label{fig:rmse_comparison_false}
\end{figure}

\subsection{Behavior of the Z-type Statistic}

In this section, we study the behavior of the Z-type statistic. To this end, we examine its ability to avoid rejecting the null hypothesis when it is true (empirical level) and its ability to reject it when it is false (empirical power). The empirical level is defined as the proportion of times the null hypothesis is rejected when it holds true, while the empirical power is the proportion of rejections when it is false. 

A total of $R=1000$ simulations were performed for both scenarios, resulting in $1000$ tests per case. The simulation setup, test conditions, and outlier generation methods remain consistent with the previous case, the null hypothesis is the one defined in \eqref{eq:hypothesis_constraint}. In the case of the Z-type statistic, it is important to note that the parameters are estimated without restrictions; subsequently, these restrictions are applied during the calculation of the Z-statistic, as shown in Definition \ref{def:def_z}.

As shown in Figure \ref{fig:empirical_level}, the test performs well in the absence of outliers, showing no significant differences across the various estimators. However, even with only a 10\% of contamination, the MLE exhibits a significantly higher rejection rate than the other estimators. In contrast, for values of $\beta = 0.8$ and $1$, the test demonstrates robustness against the presence of contamination.
\begin{figure}[H]
    \centering
     \includegraphics[width=0.9\textwidth]{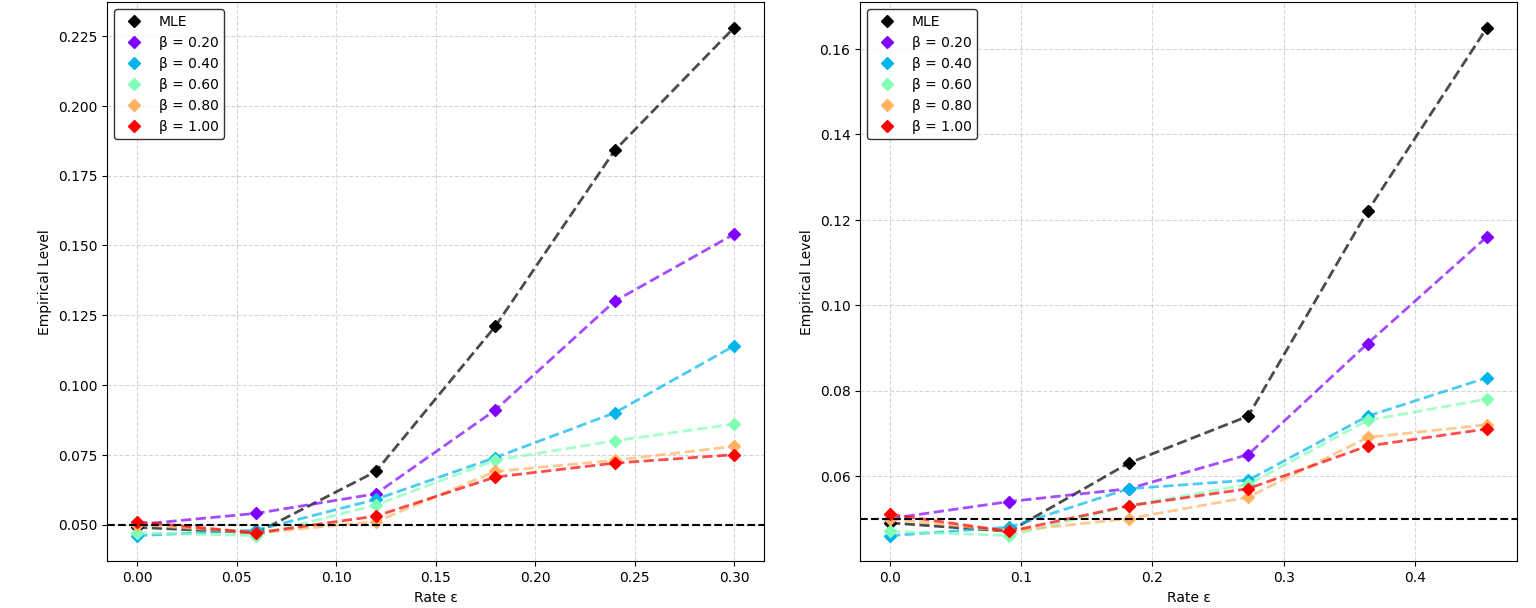}
    \caption{Empirical level of the Z-test for the two types of outliers considered (affecting $a_0$ and $a_1$).}
    \label{fig:empirical_level}
\end{figure}
Figure \ref{fig:empirical_power} displays the empirical power for each proportion of outliers. As observed, in the absence of outliers, the MLE exhibits the highest rejection rate. However, as the contamination rate increases, high values of $\beta$ continue to reject the null hypothesis at a considerable rate. In contrast, for $\beta = 0.2$ or the MLE, the power approaches nearly $5\%$; that is, the test behaves as if the null hypothesis were true.

It is necessary to specify that this outcome is not guaranteed in all scenarios. In our case, given that $a_1 = -1.4$ and outliers are included at the end of the experiment, it is as if some devices are measured to have lasted longer than expected. Consequently, the estimation of $a_1$ shifts toward a lower absolute value (approaching $a_1 = -1.15$). Nevertheless, since the true parameter is unknown in a real-world testing environment, this remains a perfectly plausible situation.
\begin{figure}[H]
    \centering
     \includegraphics[width=0.9\textwidth]{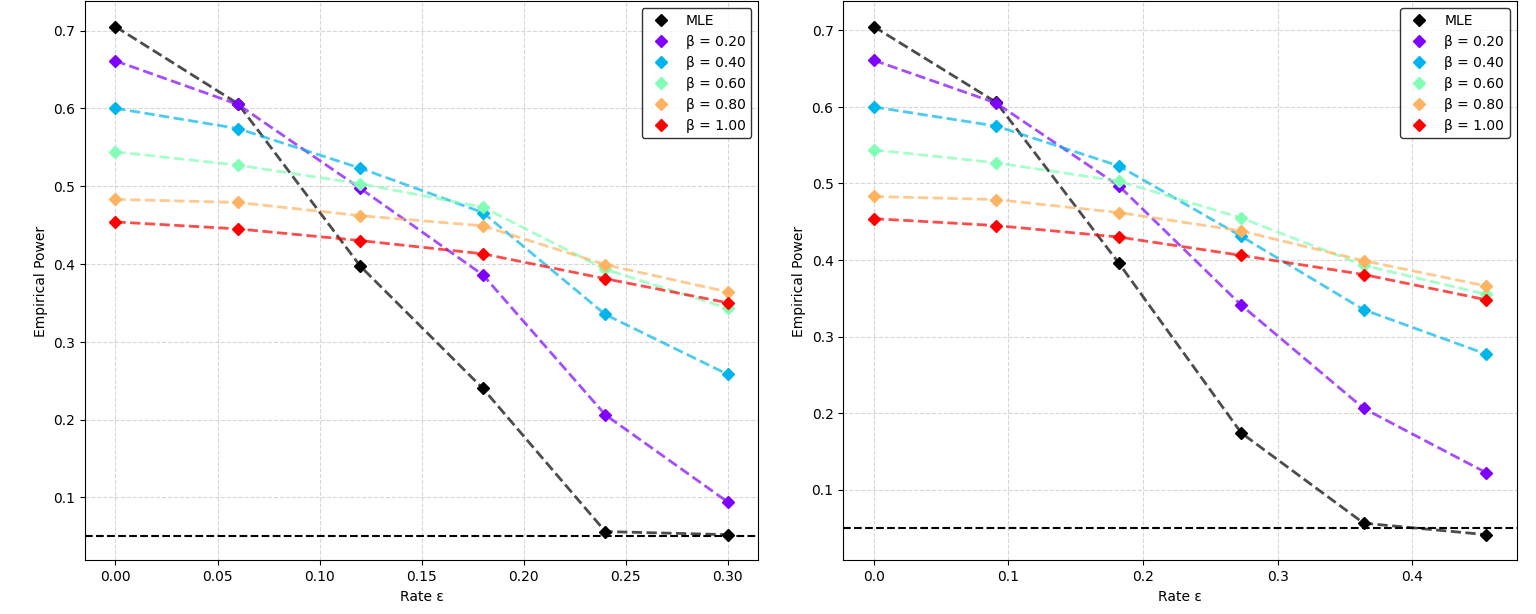}
    \caption{Empirical power of the Z-test for the two types of outliers considered (affecting $a_0$ and $a_1$).}
    \label{fig:empirical_power}
\end{figure}

Finally, an analysis was conducted as a function of the sample size. Under the same conditions, the number of devices was varied ($N=360, 720, \dots, 3600$), with $R=1000$ simulations performed for each case. Both the empirical level (true null hypothesis) and the empirical power (false null hypothesis) were obtained, considering scenarios both without outliers and with a $4\%$ outlier proportion.

Figure \ref{fig:Z_num_simulations} presents the empirical level and power as functions of the sample size under contaminated scenarios. The contamination levels are set to $\epsilon = 0.15$ for perturbations in $a_0$ and $\epsilon = 0.30$ for perturbations in $a_1$. As observed regarding the empirical level, there is no significant difference between the estimators in the absence of outliers. However, as the sample size increases, the empirical level for the MLE worsens considerably; that is, with a very large sample size, it becomes even more sensitive to the presence of outliers. As for the empirical power, the MLE performs best without outliers, though all estimators quickly converge to a power of $100\%$. In the presence of outliers, the MLE exhibits poor performance regardless of the increase in sample size, whereas the MDPDE improves significantly, especially for high values of $\beta$.
\begin{figure}[H]
    \centering
    \begin{minipage}{0.88\textwidth}
        \centering
        \begin{subfigure}[c]{0.48\linewidth}
            \centering
            $\vcenter{\hbox{\rotatebox[origin=c]{90}{\scriptsize Level}}}$%
            \hspace{4pt}%
            $\vcenter{\hbox{\includegraphics[width=0.88\linewidth]{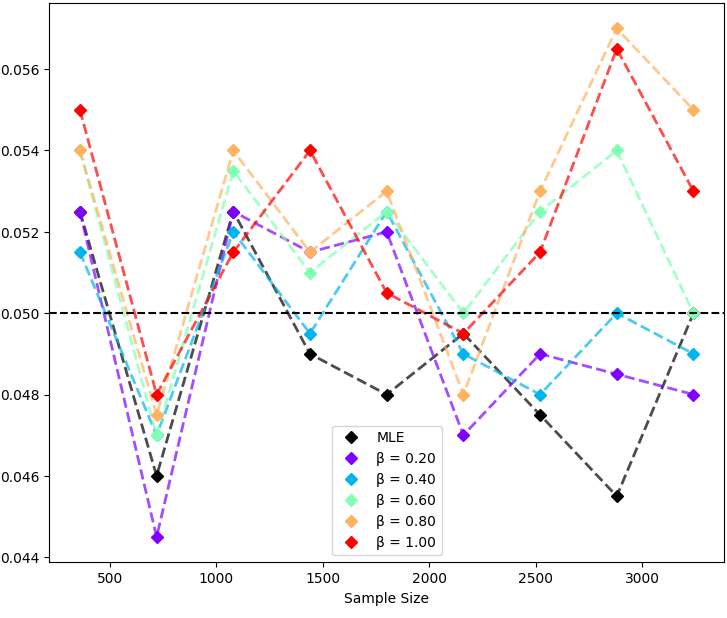}}}$
            \caption{Absence of contamination}
            \label{subfig:a1}
        \end{subfigure}\hfill
        \begin{subfigure}[c]{0.48\linewidth}
            \centering
            $\vcenter{\hbox{\rotatebox[origin=c]{90}{\scriptsize Power}}}$%
            \hspace{4pt}%
            $\vcenter{\hbox{\includegraphics[width=0.88\linewidth]{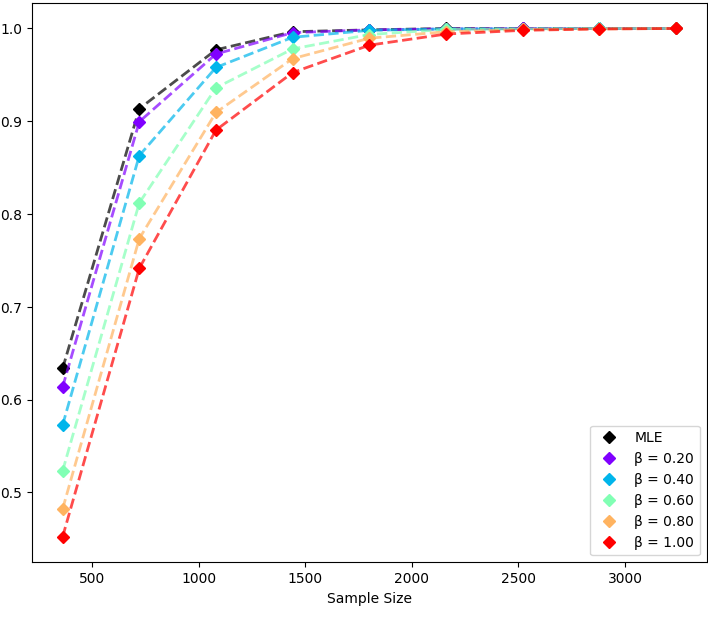}}}$
            \caption{Absence of contamination}
            \label{subfig:b1}
        \end{subfigure}

        \vspace{0.4cm}
        \begin{subfigure}[c]{0.48\linewidth}
            \centering
            $\vcenter{\hbox{\rotatebox[origin=c]{90}{\scriptsize Level}}}$%
            \hspace{4pt}%
            $\vcenter{\hbox{\includegraphics[width=0.88\linewidth]{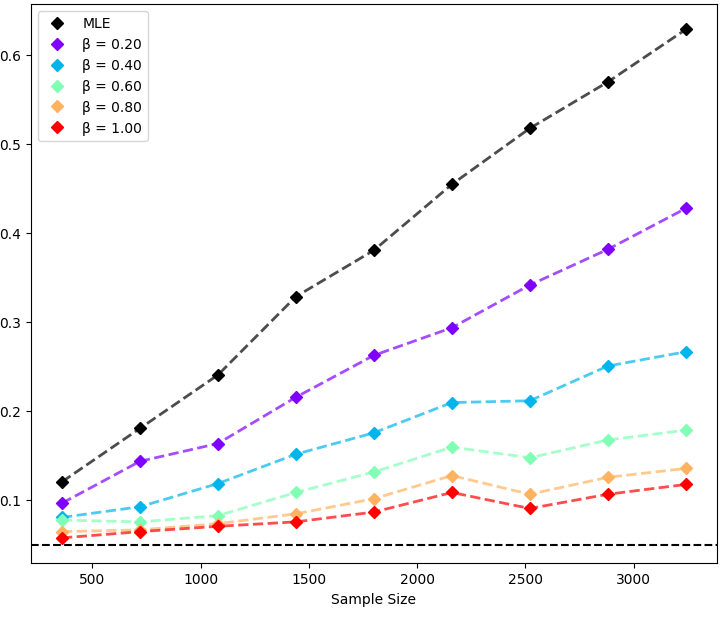}}}$
            \caption{Contamination in $a_0$}
            \label{subfig:c1}
        \end{subfigure}\hfill
        \begin{subfigure}[c]{0.48\linewidth}
            \centering
            $\vcenter{\hbox{\rotatebox[origin=c]{90}{\scriptsize Power}}}$%
            \hspace{4pt}%
            $\vcenter{\hbox{\includegraphics[width=0.88\linewidth]{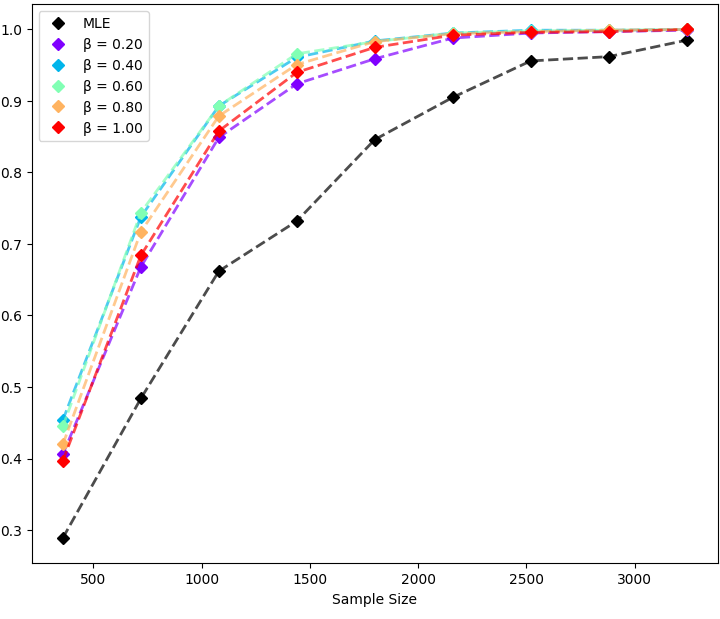}}}$
            \caption{Contamination in $a_0$}
            \label{subfig:d1}
        \end{subfigure}

        \vspace{0.4cm}
        \begin{subfigure}[c]{0.48\linewidth}
            \centering
            $\vcenter{\hbox{\rotatebox[origin=c]{90}{\scriptsize Level}}}$%
            \hspace{4pt}%
            $\vcenter{\hbox{\includegraphics[width=0.88\linewidth]{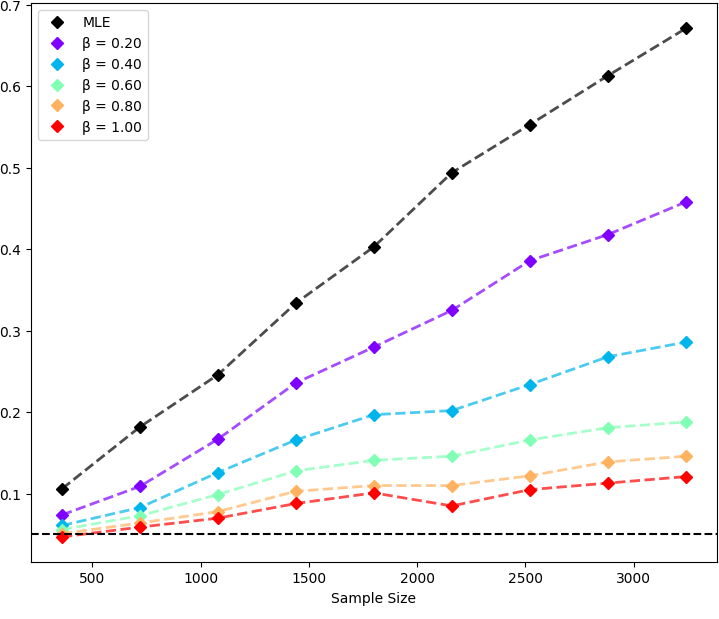}}}$
            \caption{Contamination in $a_1$}
            \label{subfig:e1}
        \end{subfigure}\hfill
        \begin{subfigure}[c]{0.48\linewidth}
            \centering
            $\vcenter{\hbox{\rotatebox[origin=c]{90}{\scriptsize Power}}}$%
            \hspace{4pt}%
            $\vcenter{\hbox{\includegraphics[width=0.88\linewidth]{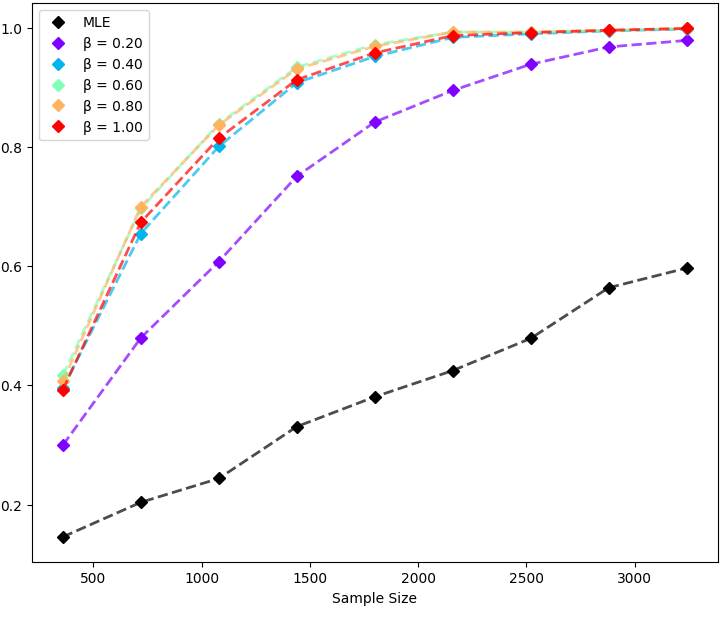}}}$
            \caption{Contamination in $a_1$}
            \label{subfig:f1}
        \end{subfigure}
    \end{minipage}
    \caption{Empirical level and power of the Z-test for the two types of outliers considered (affecting $a_0$ and $a_1$).}
    \label{fig:Z_num_simulations}
\end{figure}

\subsection{Behavior of the Rao-type Statistic}
In this section, the performance of the Rao-type statistic is analyzed. To this end, a simulation study similar to that conducted for the Z-type statistic is considered, using the same parameter settings and contamination schemes. However, two key differences must be highlighted. First, in the case of the Z-type statistic, the estimator employed is the unrestricted one, and the constraint is imposed directly within the test statistic. In contrast, the Rao-type statistic is constructed using the RMDPDE under the null hypothesis. Second, it is important to note that the Z-type statistic does not explicitly depend on the sample, whereas the Rao-type statistic does, since it is based on the expressions derived in equation \eqref{eq:score}.
As in the previous section, the empirical level and empirical power of the Rao-type statistic are analyzed under both uncontaminated and contaminated scenarios. The study follows the same simulation framework described for the Z-type statistic, allowing for a direct comparison between both approaches.

Figures \ref{fig:rao_level} and \ref{fig:rao_power} show the empirical level and power, respectively, for different contamination proportions affecting $a_0$ and $a_1$. The conclusions are consistent with those obtained for the Z-type statistic. In the absence of contamination, all estimators exhibit a satisfactory behavior, with empirical levels close to the nominal one and high power.

However, in the presence of contamination, the classical MLE becomes highly sensitive, leading to inflated rejection rates and a notable loss of power. In contrast, the MDPDE-based Rao-type statistics show a clear robustness, particularly for moderate to large values of $\beta$, maintaining stable levels and reasonable power even under contaminated scenarios.
\begin{figure}[H]
    \centering
     \includegraphics[width=0.9\textwidth]{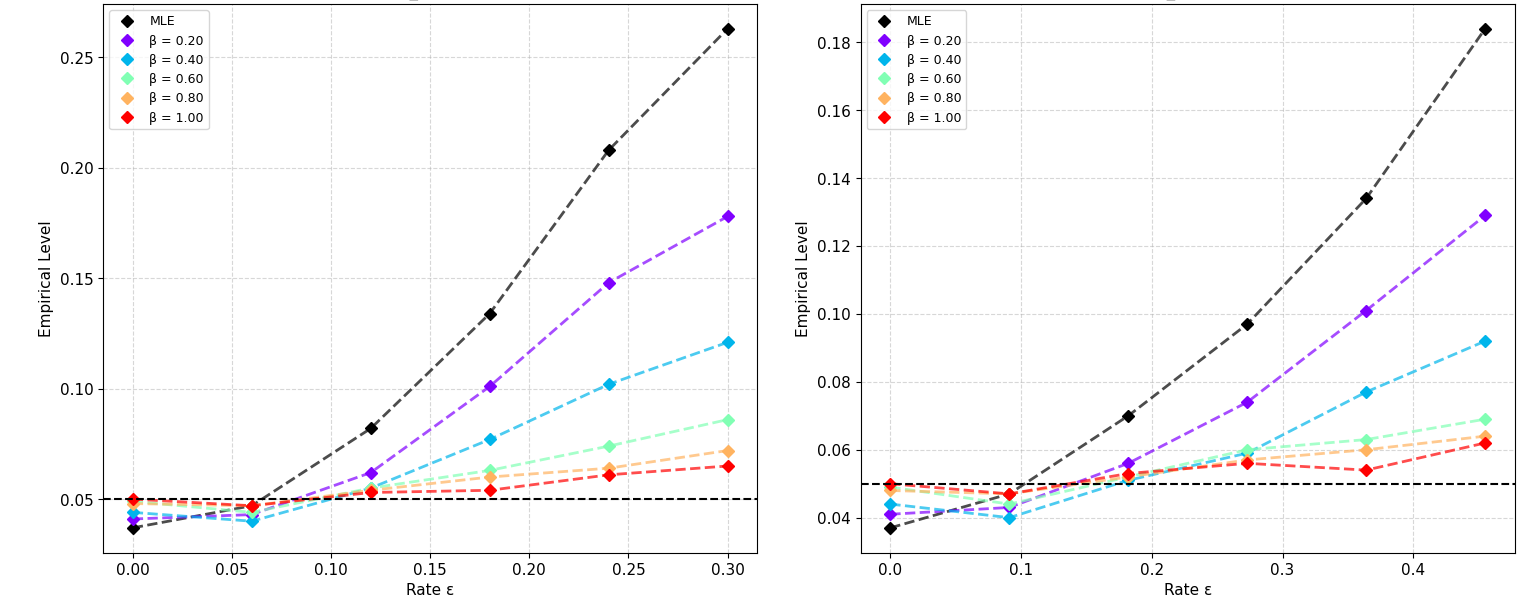}
    \caption{Empirical level of the Rao-test for the two types of outliers considered (affecting $a_0$ and $a_1$).}
    \label{fig:rao_level}
\end{figure}

\begin{figure}[H]
    \centering
     \includegraphics[width=0.9\textwidth]{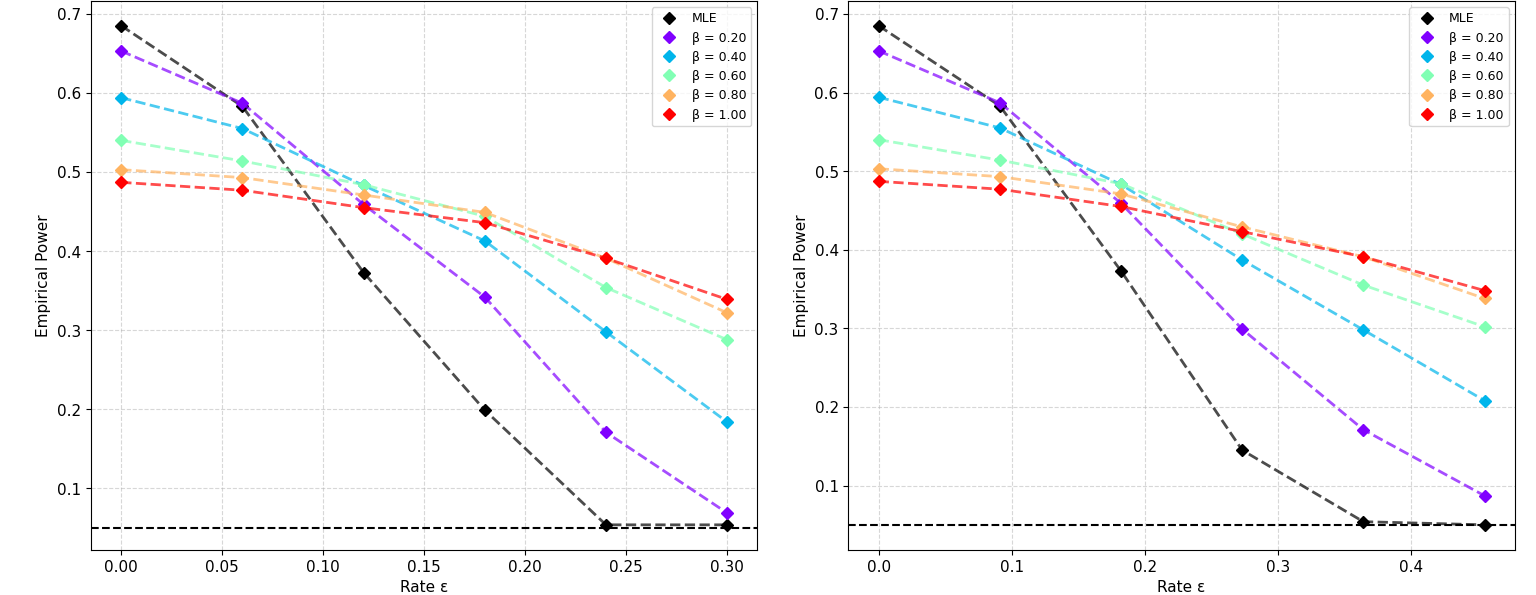}
    \caption{Empirical power of the Rao test for the two types of outliers considered (affecting $a_0$ and $a_1$).}
    \label{fig:rao_power}
\end{figure}
Finally, the effect of the sample size is also examined by computing both empirical level and power for increasing values of $N$. The results, shown in Figure \ref{fig:rao_num_simulations}, confirm the same pattern observed for the Z-type statistic: while the performance of the MLE deteriorates as the sample size increases in the presence of outliers, the robust estimators improve and tend to stabilize, achieving a better balance between level accuracy and power.
\begin{figure}[htbp]
    \centering
    \begin{minipage}{0.88\textwidth}
        \centering
        \begin{subfigure}[c]{0.48\linewidth}
            \centering
            $\vcenter{\hbox{\rotatebox[origin=c]{90}{\scriptsize Level}}}$%
            \hspace{4pt}%
            $\vcenter{\hbox{\includegraphics[width=0.88\linewidth]{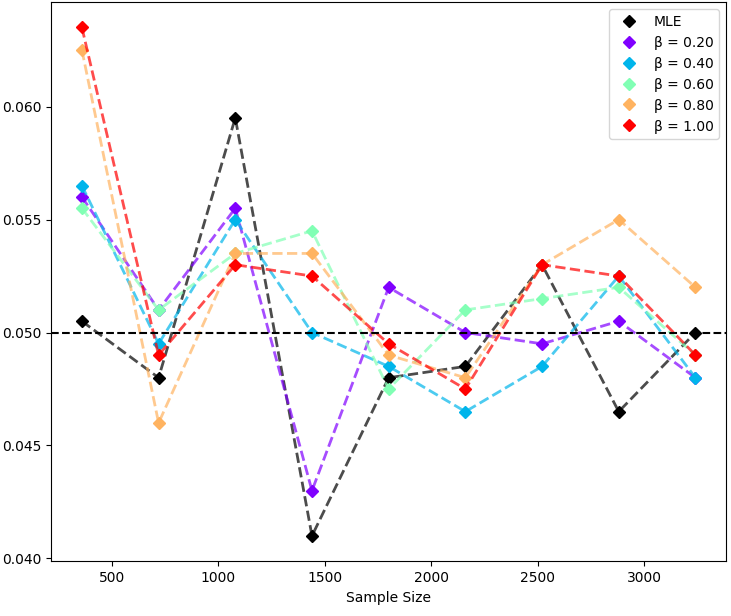}}}$
            \caption{Absence of contamination}
            \label{subfig:a2}
        \end{subfigure}\hfill
        \begin{subfigure}[c]{0.48\linewidth}
            \centering
            $\vcenter{\hbox{\rotatebox[origin=c]{90}{\scriptsize Power}}}$%
            \hspace{4pt}%
            $\vcenter{\hbox{\includegraphics[width=0.88\linewidth]{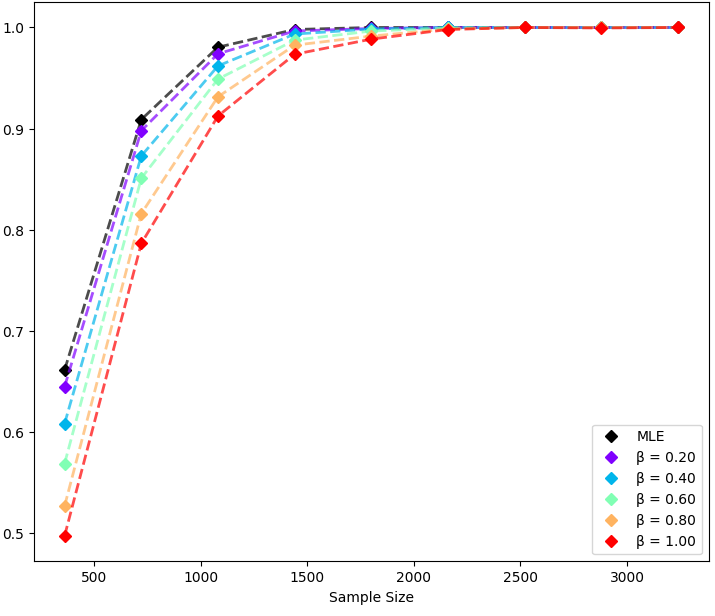}}}$
            \caption{Absence of contamination}
            \label{subfig:b2}
        \end{subfigure}

        \vspace{0.4cm}
        \begin{subfigure}[c]{0.48\linewidth}
            \centering
            $\vcenter{\hbox{\rotatebox[origin=c]{90}{\scriptsize Level}}}$%
            \hspace{4pt}%
            $\vcenter{\hbox{\includegraphics[width=0.88\linewidth]{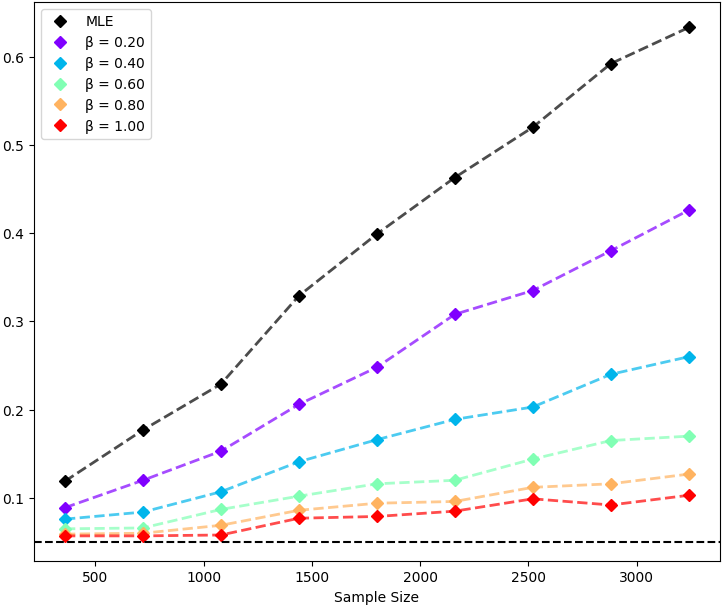}}}$
            \caption{Contamination in $a_0$}
            \label{subfig:c2}
        \end{subfigure}\hfill
        \begin{subfigure}[c]{0.48\linewidth}
            \centering
            $\vcenter{\hbox{\rotatebox[origin=c]{90}{\scriptsize Power}}}$%
            \hspace{4pt}%
            $\vcenter{\hbox{\includegraphics[width=0.88\linewidth]{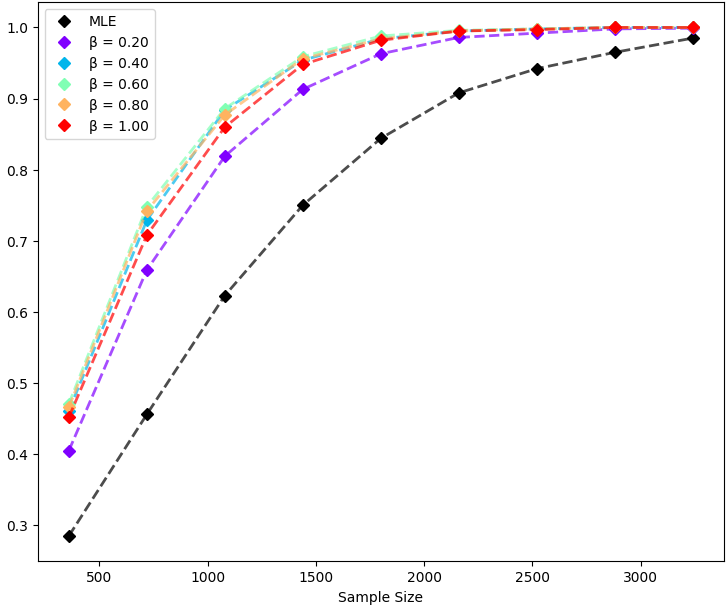}}}$
            \caption{Contamination in $a_0$}
            \label{subfig:d2}
        \end{subfigure}

        \vspace{0.4cm}
        \begin{subfigure}[c]{0.48\linewidth}
            \centering
            $\vcenter{\hbox{\rotatebox[origin=c]{90}{\scriptsize Level}}}$%
            \hspace{4pt}%
            $\vcenter{\hbox{\includegraphics[width=0.88\linewidth]{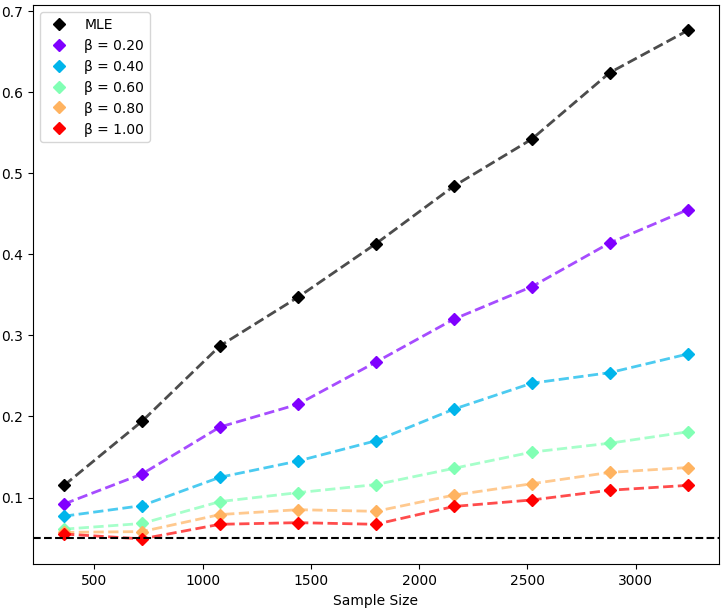}}}$
            \caption{Contamination in $a_1$}
            \label{subfig:e2}
        \end{subfigure}\hfill
        \begin{subfigure}[c]{0.48\linewidth}
            \centering
            $\vcenter{\hbox{\rotatebox[origin=c]{90}{\scriptsize Power}}}$%
            \hspace{4pt}%
            $\vcenter{\hbox{\includegraphics[width=0.88\linewidth]{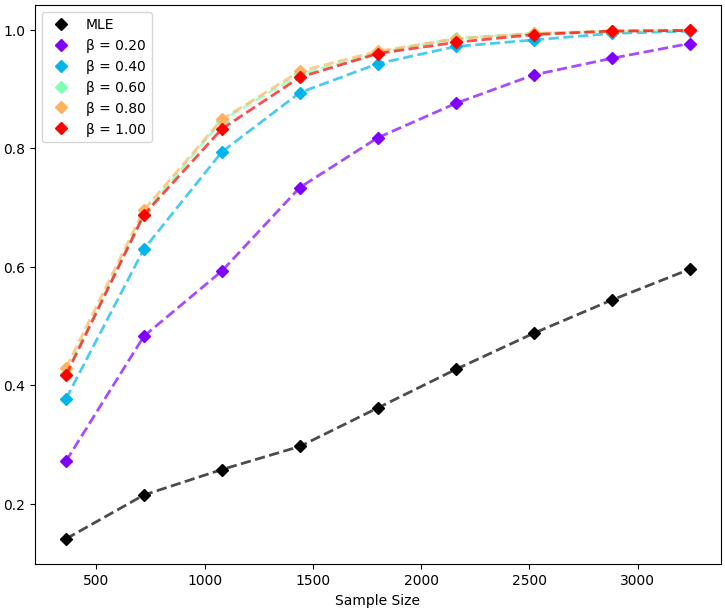}}}$
            \caption{Contamination in $a_1$}
            \label{subfig:f2}
        \end{subfigure}
    \end{minipage}
    \caption{Empirical level and power of the Rao test for the two types of outliers considered (affecting $a_0$ and $a_1$).}
    \label{fig:rao_num_simulations}
\end{figure}

\section{Real Data Analysis}\label{sec:sec8}

This section presents an analysis based on real-world data. Specifically, the failure times reported in Table 5.4 of \citet{zhu2010} are used. A SSALT was performed on lightbulbs, considering voltage as the accelerating stress factor.

The experimental setup is summarized as follows:
\begin{itemize}
    \item Sample Size ($N$): 64 devices.
    \item Total Test Duration: 140 hours.
    \item First Stress Interval ($I_1$): $[0, 96]$ hours at a voltage of 2.25 V.
    \item Second Stress Interval ($I_2$): $(96, 140]$ hours at a voltage of 2.44 V.
\end{itemize}

The observed failure times ($t_{i:N}$) from the step-voltage test are organized according to the stress intervals as follows:

During the first stress level ($V = 2.25$), a total of 34 failures  were observed; after the increase in stress level, 19 additional failures were recorded. A total of 11 devices survived the entire experiment and were right-censored at time $t = 140$.

\begin{itemize}
\item Failure Times ($t \leq 96$): 12.07, 14.00, 17.95, 19.50, 22.10, 23.11, 24.00, 24.00, 25.10, 26.46, 26.58, 26.90, 28.06, 34.00, 36.13, 36.64, 40.85, 41.11, 42.63, 44.10, 46.30, 52.51, 54.00, 58.09, 62.68, 64.17, 72.25, 73.13, 83.63, 86.90, 90.09, 91.22, 91.56, 94.38.

   \item Failure Times ($96 < t < 140$): 97.71, 101.53, 102.10, 105.10, 105.11, 109.20, 112.11, 114.40, 117.90, 119.58, 120.20, 121.90, 122.50, 123.60, 126.50, 126.95, 129.25, 130.10, 136.31.
\end{itemize}
In this case, the null hypothesis to be tested is
\[
H_0: a_1 = 0 \quad \text{versus} \quad H_1: a_1 \neq 0.
\]

That is, we want to test whether an increase in voltage affects the mean lifetime. In order to do this, both the Z-test and the Rao test, defined in Definitions \ref{def:def_z} and \ref{def:rao_test_def}, have been applied. Table \ref{tab:real_data_table} presents the results of the MDPDE and RMDPDE estimators, along with the corresponding test statistics and associated $p$-values for a significance level of $0.05$. 

\begin{table}[H]
\centering
\begin{tabular}{c c c c c c c c c}
\hline
$\beta$ & $\hat{a}_0$ & $\hat{a}_1$ & $\tilde{a}_0$ & $\tilde{a}_1$ & $Z$ & $p\text{-value}_Z$ & Rao & $p\text{-value}_{Rao}$ \\
\hline
0.0 & 17.196 & -5.474 & 4.614 & 0.0 & -3.608 & 0.000309 & 14.727 & 0.000124 \\
0.2 & 17.263& -5.510 & 4.529 & 0.0 & -3.638 & 0.000275 & 15.228 & 0.000095 \\
0.4 & 17.279 & -5.521 & 4.469 & 0.0 & -3.646 & 0.000266 & 15.398 & 0.000087 \\
0.6 & 17.190 & -5.487 & 4.431 & 0.0 & -3.616 & 0.000300 & 15.266 & 0.000093 \\
0.8 & 17.189 & -5.489 & 4.408 & 0.0 & -3.608 & 0.000308 & 14.921 & 0.000112 \\
1.0 & 17.189 & -5.490 & 4.394 & 0.0 & -3.59 & 0.000322 & 14.434 & 0.000145 \\
\hline
\end{tabular}
\caption{Estimation and Z and Rao test results}
\label{tab:real_data_table}
\end{table}
The results from both the Z-test and the Rao test consistently yield very small $p$-values across all values of $\beta$. Since all these values are well below the significance level of $0.05$, there is strong evidence to reject the null hypothesis $H_0: a_1 = 0$.Therefore, we conclude that voltage has a statistically significant effect on the mean lifetime. Furthermore, the estimates of $a_1$ are stable and negative for all values of $\beta$, suggesting a robust inverse relationship between voltage and lifetime. Finally, both testing procedures lead to consistent conclusions, supporting the reliability and robustness of the results obtained using the MDPDE and RMDPDE approaches.

\section{Conclusions}\label{sec:sec9}

In this article, we have continued the analysis of the robustness of the MDPDE within the framework of a SSALTs model under exponential lifetime distributions. Specifically, we have examined the restricted estimator, its asymptotic distribution, and hypothesis testing procedures based on the Z-type test, and the Rao-type test.

The theoretical and empirical results presented here demonstrate that the classical MLE, while optimal under ideal conditions, exhibits a marked lack of robustness in the presence of outliers or data contamination. In particular, the MLE suffers a significant deterioration in its ability to correctly control Type I and Type II errors: it tends to reject the null hypothesis when it is true, and to fail to reject it when it is false. Moreover, this degradation is not attenuated as the sample size grows; on the contrary, the empirical evidence shows that the negative effect of contamination is amplified with increasing number of observations, making the MLE increasingly unreliable in large-sample scenarios where contamination may be present. In contrast, the MDPDE, governed by the tuning parameter $\beta$, provides a flexible and effective framework for balancing statistical efficiency and robustness. For $\beta = 0$, the MDPDE reduces to the MLE, while for $\beta > 0$ it progressively downweights the influence of outlying observations, yielding estimators that remain stable under contamination. This robustness property naturally propagates to the associated hypothesis tests, which maintain their nominal significance levels and statistical power even in the presence of data contamination. Since hypothesis testing constitutes the foundation of statistical decision-making in reliability and accelerated life testing, the robustness of these tests is of paramount practical importance.

The results of this work therefore reinforce the suitability of the MDPDE-based inferential framework as a robust alternative to classical likelihood-based methods in SSALT models, particularly in applied settings where the presence of anomalous observations cannot be ruled out. Future research directions may include the extension of these results to more general lifetime distributions, the development of optimal data-driven strategies for the selection of the tuning parameter $\beta$, and the application of these robust methods to real-world accelerated life testing datasets.




\bibliographystyle{unsrtnat}
\bibliography{bibliography}  
\section{Appendix}

 \begin{pot1}
 Following from Theorem 2 \citet{Basu2018} and Theorem 2 \citet{jaenada2025} we consider 
 \begin{align*}
 h_N(\bm{a})&=\frac{1}{1+\beta}\left(h_{1}(a_0,a_1)+h_{2}(a_0,a_1)\right)
 \end{align*}
 and we have the following Taylor's expansion:
 \begin{align*}
 \left.\frac{\partial h_N(\bm{a})}{\partial \bm{a}} \right|_{\bm{a}=\tilde{\bm{a}}^\beta} 
&= \left. \frac{\partial h_N(\bm{a})}{\partial \bm{a}}  \right|_{\bm{a}=\bm{a}^*} 
+ \left. \frac{\partial^2 h_N(\bm{a})}{\partial \bm{a}^2} \right|_{\bm{a}=\bm{a}^*} 
\left(\tilde{\bm{a}}^\beta - \bm{a}^*\right)^{\intercal} + o(1) \\
&=\left. \frac{\partial h_N(\bm{a})}{\partial \bm{a}} \right|_{\bm{a}=\bm{a}^*} 
+ \bm{J}_\beta(\bm{a}^*) 
\left(\tilde{\bm{a}}^\beta - \bm{a}^*\right)^{\intercal} + o(1).
 \end{align*}
 This equality holds because
 \begin{align*}
 \frac{\partial h_N^2(\bm{a})}{\partial^2 \bm{a}}\big|_{\bm{a}=\bm{a}^*}&=\bm{J}_\beta(\bm{a}^*)\left(\tilde{\bm{a}}^\beta-\bm{a}^*\right)^{\intercal}+o(1).
 \end{align*}
 Taking expectation on the function derivative at the true parameter, we have
 \begin{align*}
 \mathbb{E}\!\left[ N^{\frac{1}{2}} 
\left. \frac{\partial h_N(\bm{a})}{\partial \bm{a}} \right|_{\bm{a}=\bm{a}^*} \right]=\bm{0}_2 & \quad \mathrm{Var}\left(N^{\frac{1}{2}} 
\left. \frac{\partial h_N(\bm{a})}{\partial \bm{a}} \right|_{\bm{a}=\bm{a}^*} \right)=\bm{K}_\beta(\bm{a}^*).
 \end{align*}
 And using the law of great numbers:
  \begin{align*}
   N^{\frac{1}{2}}\left.\frac{\partial h_N(\bm{a})}{\partial \bm{a}} \right|_{\bm{a}=\bm{a}^*} \xrightarrow[N \to \infty]{\mathcal{L}} \mathcal{N}\left(\bm{0},\bm{K}_{\beta}(\bm{a}^*)\right).
  \end{align*}
The RMDPDE is a minimum of a differentiable function with constraints. It must satisfy:
  \begin{align*}
  \begin{cases}
  N \frac{\partial h_N(\bm{a})}{\partial \bm{a}} + \bm{m} \tilde{\mu}_N&=\bm{0}_2
   \\
   \bm{m}^{\intercal}\bm{a}-d=0.
  \end{cases}
  \end{align*}
So then we have
  \begin{align*}
\left.  \frac{\partial h_N(\bm{a})}{\partial \bm{a}} \right|_{\bm{a}=\tilde{\bm{a}}}&=-N^{-1}\bm{m}\tilde{\mu}_N.
  \end{align*}
  Computing this the Taylor's expansion we have
  \begin{align*}
  \left.\frac{\partial h_N(\bm{a})}{\partial \bm{a}} \right|_{\bm{a}=\bm{a}^*} + \bm{J}_{\beta}(\bm{a}^*)\left(\tilde{\bm{a}}^{\intercal}-\bm{a}^*\right)^{\intercal}+N^{-1}\bm{m}\tilde{\mu}_N+o(1)=\bm{0}_2.
  \end{align*}
  We also have the following, multiplying by $N^{\frac{1}{2}}$ and reorganizing terms.
  \begin{align*}
  -N^{\frac{1}{2}} \left.\frac{\partial h_N(\bm{a})}{\partial \bm{a}} \right|_{\bm{a}=\bm{a}^*}&=\bm{J}_{\beta}(\bm{a}^*)N^{\frac{1}{2}}\left(\tilde{\bm{a}}^{\beta}-\bm{a}^*\right)^{\intercal}+\bm{m}N^{\frac{1}{2}}\tilde{\mu}_N.
  \end{align*}
  Noting that 
  \begin{align*}
  N^{\frac{1}{2}}\bm{m}^{\intercal}\left(\tilde{\bm{a}}^{\beta}-\bm{a}^*\right)&=N^{\frac{1}{2}}\bm{m}^{\intercal}\tilde{\bm{a}}^{\beta}-N^{\frac{1}{2}}\bm{m}^{\intercal}\bm{a}^*=N^{\frac{1}{2}}\bm{m}^{\intercal}\tilde{\bm{a}}^{\beta}-N^{\frac{1}{2}}\bm{d}
  \\
  &=N^{\frac{1}{2}}\left[\bm{m}^{\intercal}\tilde{\bm{a}}^{\beta}-d\right]=0.
  \end{align*}
  Combining the previous equations we can write
  \begin{align*}
\begin{pmatrix}
   \bm{J}_{\beta}(\bm{a}^*) & \bm{m}
   \\
   \bm{m}^{\intercal} & 0
     \end{pmatrix}	
  \begin{pmatrix}
  N^{\frac{1}{2}}\left(\tilde{\bm{a}}^{\beta}-\bm{a}^*\right)
  \\
  N^{-\frac{1}{2}}\tilde{\mu}_N
    \end{pmatrix}
    =
  \begin{pmatrix}
   -N^{\frac{1}{2}} \left.\frac{\partial h_N(\bm{a})}{\partial \bm{a}} \right|_{\bm{a}=\bm{a}^*}
   \\
   0
 \end{pmatrix},
  \end{align*}
  and inverting the first matrix by blocks
  \begin{align}\label{eq:eqmat}
  \begin{pmatrix}
   \bm{J}_{\beta}(\bm{a}^*) & \bm{m}
   \\
   \bm{m}^{\intercal} & 0
     \end{pmatrix}	^{-1}
     &=
   \begin{pmatrix}
   \bm{J}_{\beta}(\bm{a}^*)^{-1}+\bm{J}_{\beta}(\bm{a}^*)^{-1}\bm{m}\left(-\bm{m}^{\intercal}\bm{J}_{\beta}(\bm{a}^*)\bm{m}\right)^{-1}\bm{m}^{\intercal}\bm{J}_{\beta}(\bm{a}^*)^{-1} & \bm{J}_{\beta}(\bm{a}^*)^{-1}\bm{m}\left(\bm{m}^{\intercal}\bm{J}_{\beta}(\bm{a}^*)\bm{m}\right)^{-1}
   \\
   \left(\bm{m}^{\intercal}\bm{J}_{\beta}(\bm{a}^*)\bm{m}\right)^{-1}\bm{m}^{\intercal}\bm{J}_{\beta}(\bm{a}^*)^{-1}&   \left(\bm{m}^{\intercal}\bm{J}_{\beta}(\bm{a}^*)\bm{m}\right)^{-1}
     \end{pmatrix}  
     \\
     &=  
        \begin{pmatrix}
        \bm{P}_{\beta}(\bm{a}^*) &    \bm{Q}_{\beta}(\bm{a}^*) 
   \\
      \bm{Q}_{\beta}(\bm{a}^*) &    \bm{R}_{\beta}(\bm{a}^*) 
     \end{pmatrix}.
  \end{align}
  Finally, by the asymptotic convergence of $N^{\frac{1}{2}} \left.\frac{\partial h_N(\bm{a})}{\partial \bm{a}} \right|_{\bm{a}=\bm{a}^*}$ the results holds.
 \end{pot1}

 \begin{pot2}
 Under the null hypothesis we have:
 \begin{align*}
 \bm{m}^{\intercal}\hat{\bm{a}}^{\beta}-d=\bm{m}^{\intercal}\left(\hat{\bm{a}}^{\beta}-\bm{a}^{*}\right).
 \end{align*}
 Then, from the Theorem 2 \citet{jaenada2025}:
 \begin{align*}
 \sqrt{N}\left(\tilde{\bm{a}}^\beta - \bm{a}^*\right) \xrightarrow[N \to \infty]{\mathcal{L}} \mathcal{N}\left(\bm{0},\bm{J}_{\beta}\left(\hat{\bm{a}}^{\beta}\right)^{-1}\bm{K}_{\beta}\left(\hat{\bm{a}}^{\beta}\right)\bm{J}_{\beta}\left(\hat{\bm{a}}^{\beta}\right)^{-1}\right)
 \end{align*}
and evaluating at $\hat{\bm{a}}^{\beta}$ we have:
 \begin{align*}
\sqrt{N}\left(\bm{m}^{\intercal}\bm{J}_{\beta}\left(\hat{\bm{a}}^{\beta}\right)^{-1}\bm{K}_{\beta}\left(\hat{\bm{a}}^{\beta}\right)\bm{J}_{\beta}\left(\hat{\bm{a}}^{\beta}\right)^{-1}\bm{m}\right)^{-\frac{1}{2}}\left(\bm{m}^{\intercal}\hat{\bm{a}}^{\beta}-d\right). \xrightarrow[N \to \infty]{\mathcal{L}} \mathcal{N}\left(1,0\right)
 \end{align*}
 as $\hat{\bm{a}}^{\beta}$ is a consistent estimator of $\bm{a}^{*}$, the stated result follows from Slutsky's theorem.
\end{pot2}

\begin{pot3}
The power function is the probability of rejection:
\begin{align*}
B_{N}\left(\bm{a}^*\right) &=P\left(\left|Z_{N}\left(\hat{\bm{a}}^{\beta} \right) \right| \ge z_{\alpha/2} \left| \bm{a}=\bm{a}^* \right. \right)
\\
&=P\left(Z_{N}\left(\hat{\bm{a}}^{\beta} \right)  \le -z_{\alpha/2} \left| \bm{a}=\bm{a}^*  \right. \right) +P\left(Z_{N}\left(\hat{\bm{a}}^{\beta} \right)  \lg z_{\alpha/2} \left| \bm{a}=\bm{a}^*  \right. \right)
\\
&=P\left(
    \sqrt{\frac{N}{\bm{m}^{\intercal}\bm{J}_{\beta}\left(\hat{\bm{a}}_L\right)^{-1}\bm{K}_{\beta}\left(\hat{\bm{a}}_L\right)\bm{J}_{\beta}\left(\hat{\bm{a}}_L\right)^{-1}\bm{m}}} \bm{m}^{\intercal}{\left(\hat{\bm{a}}^{\beta}-\bm{a}^*\right)}
    \right.
  \\
  & \left.  \le -z_{\alpha/2} -\sqrt{\frac{N}{\bm{m}^{\intercal}\bm{J}_{\beta}\left(\hat{\bm{a}}_L\right)^{-1}\bm{K}_{\beta}\left(\hat{\bm{a}}_L\right)\bm{J}_{\beta}\left(\hat{\bm{a}}_L\right)^{-1}\bm{m}}}\left(\bm{m}^{\intercal} \bm{a}^*-d\right)
\right)
\\
+&1-P\left(
    \sqrt{\frac{N}{\bm{m}^{\intercal}\bm{J}_{\beta}\left(\hat{\bm{a}}_L\right)^{-1}\bm{K}_{\beta}\left(\hat{\bm{a}}_L\right)\bm{J}_{\beta}\left(\hat{\bm{a}}_L\right)^{-1}\bm{m}}} \bm{m}^{\intercal}{\left(\hat{\bm{a}}^{\beta}-\bm{a}^*\right)}
    \right.
  \\
  & \left.  \ge z_{\alpha/2} -\sqrt{\frac{N}{\bm{m}^{\intercal}\bm{J}_{\beta}\left(\hat{\bm{a}}_L\right)^{-1}\bm{K}_{\beta}\left(\hat{\bm{a}}_L\right)\bm{J}_{\beta}\left(\hat{\bm{a}}_L\right)^{-1}\bm{m}}}\left(\bm{m}^{\intercal} \bm{a}^*-d\right),
\right)
\end{align*}
as $\hat{\bm{a}}^{\beta} \xrightarrow{p}\bm{a}^*$ and $\sqrt{N}\left(\hat{\bm{a}}^{\beta}-\bm{a}^*\right)\xrightarrow[N \to \infty]{\mathcal{L}}\mathcal{N}\left(\bm{0},\bm{J}_{\beta}\left(\hat{\bm{a}}_L\right)^{-1}\bm{K}_{\beta}\left(\hat{\bm{a}}_L\right)\bm{J}_{\beta}\left(\hat{\bm{a}}_L\right)^{-1}\right)$, the results follows from Slutsky's theorem.
\end{pot3}

\begin{pot4}
We can rewrite
\begin{align*}
\bm{m}^{\intercal}\hat{\bm{a}}^{\beta}-d&=\bm{m}^{\intercal}\bm{a}_L-d+\bm{m}^{\intercal}\left(\hat{\bm{a}}^{\beta}-\bm{a}_L\right)
\\
&=\frac{1}{\sqrt{N}}\bm{m}^{\intercal}\bm{\ell}+\bm{m}^{\intercal}\left(\hat{\bm{a}}^{\beta}-\bm{a}_L\right),
\end{align*}
and then 
\begin{align*}
\sqrt{N}\left(\bm{m}^{\intercal}\hat{\bm{a}}^{\beta}-d\right)=\bm{m}^{\intercal}\bm{\ell}+\bm{m}^{\intercal}\sqrt{N}\left(\hat{\bm{a}}^{\beta}-\bm{a}_L\right).
\end{align*}
From Theorem 2 \citet{jaenada2025} we have:
\begin{align*}
 \sqrt{N}\left(\tilde{\bm{a}}^\beta - \bm{a}_L\right) \xrightarrow[N \to \infty]{\mathcal{L}} \mathcal{N}\left(\bm{0},\bm{J}_{\beta}\left(\hat{\bm{a}}_L\right)^{-1}\bm{K}_{\beta}\left(\hat{\bm{a}}_L\right)\bm{J}_{\beta}\left(\hat{\bm{a}}_L\right)^{-1}\right).
\end{align*}
Then
\begin{align*}
 \sqrt{N}\left(\bm{m}^{\intercal}\hat{\bm{a}}^\beta - d\right) \xrightarrow[N \to \infty]{\mathcal{L}} \mathcal{N}\left(\bm{m}^{\intercal}\bm{\ell},\bm{m}^{\intercal}\bm{J}_{\beta}\left(\hat{\bm{a}}_L\right)^{-1}\bm{K}_{\beta}\left(\hat{\bm{a}}_L\right)\bm{J}_{\beta}\left(\hat{\bm{a}}_L\right)^{-1}\bm{m}\right).
\end{align*}
So
\begin{align*}
\frac{\sqrt{N}\left(\bm{m}^{\intercal}\hat{\bm{a}}^{\beta}-d\right)-\bm{m}^{\intercal}\bm{\ell}}{\sqrt{\bm{m}^{\intercal}\bm{J}_{\beta}\left(\hat{\bm{a}}_L\right)^{-1}\bm{K}_{\beta}\left(\hat{\bm{a}}_L\right)\bm{J}_{\beta}\left(\hat{\bm{a}}_L\right)^{-1}\bm{m}}} \xrightarrow[N \to \infty]{\mathcal{L}} \mathcal{N}\left(0,1\right).
\end{align*}
As $\hat{\bm{a}}^{\beta} \xrightarrow{p} \bm{a}_L$, the result follows from Slutsky's theorem evaluating the variance at $\hat{\bm{a}}^{\beta}$.
\end{pot4}

\begin{pot5}
It is shown in proof of Theorem \ref{th:theo_2} that
 \begin{align*}
 \mathbb{E}\!\left[ N^{\frac{1}{2}} 
\left. \frac{\partial h_N(\bm{a})}{\partial \bm{a}} \right|_{\bm{a}=\bm{a}^*} \right]=\bm{0}_2 & \quad \mathrm{Var}\left(N^{\frac{1}{2}} 
\left. \frac{\partial h_N(\bm{a})}{\partial \bm{a}} \right|_{\bm{a}=\bm{a}^*} \right)=\bm{K}_\beta(\bm{a}^*).
 \end{align*}
 So the result follows from the asymptotic theory.
\end{pot5}

\begin{pot6}
The RMDPDE, $\tilde{\bm{a}}^{\beta}$ is the minimum of the DPD loss restricted to $\bm{m}^{\intercal}\bm{a}=d$. It satisfies the lagrangian multipliers equations
\begin{align*}
\left. \frac{\partial h_N(\bm{a})}{\partial \bm{a}} \right|_{\bm{a}=\tilde{\bm{a}}^{\beta}}-\bm{m}\tilde{\bm{\mu}} =U_{\beta,N}\left(\tilde{\bm{a}}\right)-\bm{m}\tilde{\bm{\mu}}=0.
\end{align*}
Then we can write 
\begin{align}\label{eq:eqUblam}
-U_{\beta,N}\left(\tilde{\bm{a}}\right)=\bm{m}\tilde{\bm{\mu}}.
\end{align}
So:
\begin{align*}
U_{\beta,N}\left(\tilde{\bm{a}}\right)^{\intercal}Q_{\beta}\left(\tilde{\bm{a}}^{\beta}\right)&=\tilde{\bm{\mu}}\bm{m}^{\intercal}Q_{\beta}\left(\tilde{\bm{a}}^{\beta}\right)
\\
&=\tilde{\bm{\mu}}\bm{m}^{\intercal}J_{\beta}\left(\tilde{\bm{a}}^{\beta}\right)^{-1}\bm{m}\left(\bm{m}^{\intercal}J_{\beta}\left(\tilde{\bm{a}}^{\beta}\right)^{-1}\bm{m}\right)^{-1}=\tilde{\bm{\mu}}.
\end{align*}
The Rao-type statistic can then be rewritten as:
\begin{align*}
R_{\beta,N}\left(\tilde{\bm{a}}^{\beta}\right)&=N\, U_{\beta,N}\left(\tilde{\bm{a}}^{\beta}\right)^{\intercal}Q_{\beta}\left(\tilde{\bm{a}}^{\beta}\right)\left[ Q_{\beta}\left(\tilde{\bm{a}}^{\beta}\right)^{\intercal}K_{\beta}\left(\tilde{\bm{a}}^{\beta}\right)Q_{\beta}\left(\tilde{\bm{a}}^{\beta}\right)\right]^{-1}Q_{\beta}\left(\tilde{\bm{a}}^{\beta}\right)^{\intercal}U_{\beta}\left(\tilde{\bm{a}}^{\beta}\right)
\\
&=N \tilde{\bm{\mu}}\left[ Q_{\beta}\left(\tilde{\bm{a}}^{\beta}\right)^{\intercal}K_{\beta}\left(\tilde{\bm{a}}^{\beta}\right)Q_{\beta}\left(\tilde{\bm{a}}^{\beta}\right)\right]^{-1}\tilde{\bm{\mu}}.
\end{align*}
Let us find the asymptotic distribution of $\tilde{\bm{\mu}}$ in order to obtain the asymptotic distribution of the Rao-test statistic. The second order Taylor expansion of the function $U_{\beta,N}\left(\bm{a}\right)$ around the true parameter value $\bm{a}^*$ is:
\begin{align*}
U_{\beta,N}\left(\tilde{\bm{a}}^{\beta}\right)=U_{\beta,N}\left(\bm{a}^*\right)+\left.\frac{\partial U_{\beta,N}\left(\bm{a}\right)}{\partial \bm{a}}\right|_{\bm{a}=\bm{a}^*}\left(\tilde{\bm{a}}^{\beta}-\bm{a}^*\right)+o\left(\left\|\tilde{\bm{a}}^{\beta}-\bm{a}^*\right\|\bm{1}_{2}\right).
\end{align*}
And since the proof of \ref{th:theo_1} 
\begin{align*}
\left.\frac{\partial U_{\beta,N}\left(\bm{a}\right)}{\partial \bm{a}}\right|_{\bm{a}=\bm{a}^*}=\left.\frac{\partial^2 h_{N}\left(\bm{a}\right)}{\partial \bm{a}^2}\right|_{\bm{a}=\bm{a}^*}=\bm{J}_{\beta}\left(\bm{a}^*\right)+o(1).
\end{align*}
So:
\begin{align*}
\left.\frac{\partial U_{\beta,N}\left(\bm{a}\right)}{\partial \bm{a}}\right|_{\bm{a}=\bm{a}^*} \xrightarrow[N \to \infty]{p} \bm{J}_{\beta}\left(\bm{a}^*\right).
\end{align*}
Therefore,
\begin{align*}
U_{\beta,N}\left(\tilde{\bm{a}}^{\beta}\right)=U_{\beta,N}\left(\bm{a}^*\right)+\bm{J}_{\beta}\left(\bm{a}^*\right)\left(\tilde{\bm{a}}^{\beta}-\bm{a}^*\right)+o\left(\left\|\tilde{\bm{a}}^{\beta}-\bm{a}^*\right\|\bm{1}_{2}\right)+o(1).
\end{align*}
Taking into account:
\begin{align*}
\bm{m}^{\intercal}\tilde{\bm{a}}^{\beta}-d=\bm{m}^{\intercal}\left(\tilde{\bm{a}}^{\beta}-\bm{a}^*\right)=0,
\end{align*}
and using equation \eqref{eq:eqUblam} we rewrite
\begin{align*}
-U_{\beta,N}\left(\bm{a}^*\right)+o\left(\left\|\tilde{\bm{a}}^{\beta}-\bm{a}^*\right\|\bm{1}_{2}\right)+o(1)=-U_{\beta,N}\left(\tilde{\bm{a}}^{\beta}\right)+\bm{J}_{\beta}\left(\bm{a}^*\right)\left(\tilde{\bm{a}}^{\beta}-\bm{a}^*\right)=\bm{m}\tilde{\bm{\mu}}+\bm{J}_{\beta}\left(\bm{a}^*\right)\left(\tilde{\bm{a}}^{\beta}-\bm{a}^*\right)
\end{align*}
We can write the two equations above in matrix form:
\begin{align*}
\begin{pmatrix}
   \bm{J}_{\beta}(\bm{a}^*) & \bm{m}
   \\
   \bm{m}^{\intercal} & 0
     \end{pmatrix}	
     \begin{pmatrix}
   \tilde{\bm{a}}^{\beta}-\bm{a}^*
   \\
   \tilde{\bm{\mu}}
     \end{pmatrix}	
     =
          \begin{pmatrix}
  -U_{\beta,N}\left(\bm{a}^*\right)
   \\
  0
     \end{pmatrix}	
     +
          \begin{pmatrix}
   o\left(\left\|\tilde{\bm{a}}^{\beta}-\bm{a}^*\right\|\bm{1}_{2}\right)+o(1)
   \\
  0
     \end{pmatrix}.
\end{align*}
and solving the equation we have
\begin{align*}
\begin{pmatrix}
   \bm{J}_{\beta}(\bm{a}^*) & \bm{m}
   \\
   \bm{m}^{\intercal} & 0
     \end{pmatrix}	^{-1}
\left(
          \begin{pmatrix}
  -U_{\beta,N}\left(\bm{a}^*\right)
   \\
  0
     \end{pmatrix}	
     +
          \begin{pmatrix}
   o\left(\left\|\tilde{\bm{a}}^{\beta}-\bm{a}^*\right\|\bm{1}_{2}\right)+o(1)
   \\
  0
     \end{pmatrix}	
\right)
     =
          \begin{pmatrix}
   \tilde{\bm{a}}^{\beta}-\bm{a}^*
   \\
   \tilde{\bm{\mu}}
     \end{pmatrix}	.
\end{align*}

We have from \eqref{eq:eqmat}:
  \begin{align*}
  \begin{pmatrix}
   \bm{J}_{\beta}(\bm{a}^*) & \bm{m}
   \\
   \bm{m}^{\intercal} & 0
     \end{pmatrix}	^{-1}
     &=  
        \begin{pmatrix}
        \bm{P}_{\beta}(\bm{a}^*) &    \bm{Q}_{\beta}(\bm{a}^*) 
   \\
      \bm{Q}_{\beta}(\bm{a}^*) &    \bm{R}_{\beta}(\bm{a}^*) 
     \end{pmatrix}.
  \end{align*}
  From theorem \ref{th:theo_5} we have the following asymptotic distribution
  \begin{align*}
            \begin{pmatrix}
  \sqrt{N}U_{\beta,N}\left(\bm{a}^*\right)
   \\
  0
     \end{pmatrix}	
     \xrightarrow[N \to \infty]{\mathcal{L}} 
     \mathcal{N}\left(\bm{0}_3,
             \begin{pmatrix}
        \bm{K}_{\beta}(\bm{a}^*) &    \bm{0}
   \\
      \bm{0}^{\intercal} &    \bm{0}
     \end{pmatrix}
     \right).
  \end{align*}
  We have then
    \begin{align*}
            \begin{pmatrix}
  \sqrt{N}\left(\tilde{\bm{a}}^{\beta}-\bm{a}^*\right)
   \\
  \sqrt{N} \tilde{\bm{\mu}}
     \end{pmatrix}	
     \xrightarrow[N \to \infty]{\mathcal{L}} 
     \mathcal{N}\left(\bm{0}_3,\bm{V}_{\beta}\left(\bm{a}^*\right)
     \right).
  \end{align*}
 With
 \begin{align*}
 \bm{V}_{\beta}\left(\bm{a}^*\right)=
        \begin{pmatrix}
        \bm{P}_{\beta}(\bm{a}^*) &    \bm{Q}_{\beta}(\bm{a}^*) 
   \\
      \bm{Q}_{\beta}(\bm{a}^*) &    \bm{R}_{\beta}(\bm{a}^*) 
     \end{pmatrix}
                  \begin{pmatrix}
        \bm{K}_{\beta}(\bm{a}^*) &    \bm{0}
   \\
      \bm{0}^{\intercal} &    \bm{0}
     \end{pmatrix}
            \begin{pmatrix}
        \bm{P}_{\beta}(\bm{a}^*) &    \bm{Q}_{\beta}(\bm{a}^*) 
   \\
      \bm{Q}_{\beta}(\bm{a}^*) &    \bm{R}_{\beta}(\bm{a}^*) 
     \end{pmatrix}.
 \end{align*}
 Then, the asymptotic distribution of $\sqrt{N}\tilde{\bm{\mu}}$ is
   \begin{align*}
   \sqrt{N}\tilde{\bm{\mu}} \xrightarrow[N \to \infty]{\mathcal{L}} \mathcal{N}\left(0,\bm{Q}_{\beta}(\bm{a}^*)^{\intercal} \bm{K}_{\beta}(\bm{a}^*) \bm{Q}_{\beta}(\bm{a}^*)\right).
   \end{align*}
   Using the convergence and consistency of RMDPDE it follows the distribution of $R_{\beta,N}\left(\tilde{\bm{a}}^{\beta}\right)$ because
   \begin{align*}
    \sqrt{N}\tilde{\bm{\mu}}\left[ \bm{Q}_{\beta}\left(\tilde{\bm{a}}^{\beta}\right)^{\intercal}\bm{K}_{\beta}\left(\tilde{\bm{a}}^{\beta}\right)\bm{Q}_{\beta}\left(\tilde{\bm{a}}^{\beta}\right)\right]^{-\frac{1}{2}} \xrightarrow[N \to \infty]{\mathcal{L}} \mathcal{N}\left(0,1\right).  
   \end{align*}
\end{pot6}

\end{document}